\documentclass[journal,draftcls,onecolumn,12pt,twoside]{IEEEtranTCOM}
\normalsize
\ifCLASSINFOpdf
\else
\fi
\date{}
\usepackage{amsmath,amssymb,amsfonts}
\usepackage{mathrsfs}
\usepackage{graphicx}
\usepackage{textcomp}
\usepackage{xcolor}
\usepackage[T1]{fontenc}
\usepackage[utf8]{inputenc}
\usepackage{authblk}
\usepackage{cite}
\usepackage{enumerate}
\usepackage{multirow}
\usepackage{epsfig,rotating,setspace,latexsym,amsmath,epsf,amssymb,bm,color}
\usepackage{arydshln}
\allowdisplaybreaks
\newtheorem{Theo}{Theorem}
\newtheorem{Lem}{Lemma}

\newtheorem{Def}{Definition}
\newtheorem{Remark}{Remark}
\newtheorem{Proposition}{Proposition}
\newtheorem{proof}{Proof}

\begin{document}
\title{Three-user D2D Coded Caching with Two Random Requesters and One Sender }

\author{Wuqu~Wang,
	Nan~Liu,~\IEEEmembership{Member,~IEEE,}
	and Wei~Kang,~\IEEEmembership{Member,~IEEE}
	
	\thanks{The authors are with the National Mobile Communications Research Laboratory, Southeast University, Nanjing 210096, China (e-mail: \{wuquwang, nanliu, wkang\}@seu.edu.cn).}
	
}

\maketitle

\begin{abstract}

In device-to-device (D2D) coded caching problems, it is possible that not all users will make file requests in the delivery phase. Hence,  we propose a new D2D centralized coded caching problem, named \emph{the 3-user D2D coded caching with two random requesters and one sender (2RR1S)}, where in the delivery phase, any two of the three users will make file requests, and the user that does not make any file request is the designated sender. 
We find the optimal caching and delivery scheme, denoted as \emph{the 2RRIS scheme}, for any number of files $N$ by proving matching converse and achievability results. It is shown that coded cache placement is needed to achieve the optimal performance. Furthermore, the optimal rate-memory tradeoff has a uniform expression for $N \geq 4$ and different expressions for $N=2$ and $3$.

To examine the usefulness of the proposed model and scheme, we adapt the 2RR1S scheme to three scenarios. The first one is the \emph{3-user D2D coded caching} model proposed by Ji \emph{et al}. By characterizing the optimal rate-memory tradeoff for the 3-user D2D coded caching when $N=2$,  which was previously unknown, we show that the adapted 2RR1S scheme is in fact optimal for the 3-user D2D coded caching problem when $N=2$ and the cache size is medium. The benefit comes from \emph{coded} cache placement which is missing from existing D2D coded caching schemes. The second scenario is where in the delivery phase, each user makes a file request randomly and independently with the same probability $p$. We call this model \emph{the request-random D2D coded caching} problem. Adapting the 2RR1S scheme to this scenario, we show the superiority of our adapted scheme over other existing D2D coded caching schemes for medium to large cache size. The third scenario is the \emph{K-user D2D coded caching with $K-s$ random requesters and $s$ senders} problem, for which an achievability result is obtained by generalizing the 2RR1S scheme.
\end{abstract}


\section{Introduction}

The applications of wireless networks have developed from traditional real-time voice communication to multimedia transmissions such as video, virtual/augmented reality game, high definition map etc., which require the throughput of each user to increase by nearly 1000 times \cite{6191306}. Fortunately, such content can be pre-stored into the user's storage during periods of low network utilization, thus avoiding network congestion during peak hours. This technology is known as \emph{caching} \cite{6495773}, \cite{6871674}. The caching process is typically divided into two phases \cite{6620392}. The \emph{placement phase} happens during the off-peak hours, where the server fills the users’ caches before the users request any content, while the \emph{delivery phase} represents the transmission stage of the server when the users reveal their demands during peak hours. Caching technology has developed rapidly in recent years, and it is currently considered as one of the effective solutions to relieve the load pressure of wireless networks. 

In traditional caching, the users cache the most likely requested contents, and the server transmits the uncached portions of the files requested by the users. 
Both the cached contents of the users and the transmitted signal of the server are uncoded. 
Contrary to traditional caching, Maddah-Ali and Niesen proposed an idea \cite{6620392} of combining coded multi-casting and device caching to satisfy multiple uni-cast demands simultaneously through coded multi-cast transmissions, which is known as \emph{coded caching}. 
The coded caching problem allows both coded cache contents of the users and coded transmission from the server. 
The goal is to design a caching and delivery scheme such that the worst-case delivery rate is the smallest, where ``worst-case'' refers to the largest delivery rate among all possible request demands of the users.  When the optimal caching and delivery scheme that achieves the smallest worst-case delivery rate can be identified for any cache size of the users, the optimal rate-memory tradeoff is found for the system. 
If each user directly stores a subset of the files' bits without coding, then the cache placement scheme is called \emph{uncoded}, otherwise it is called \emph{coded}. The coded caching problem studied in \cite{6620392} is of a \emph{centralized} nature, where it is assumed that the set of users present during the placement phase will each request a file at the beginning of the delivery phase. \emph{Decentralized coded caching} has been studied in \cite{6807823}, which considers the possibility that some users may leave or turn off during the delivery phase, and studies less coordinated caching strategies.


To further reduce the traffic load of the server at peak hours, Ji \emph{et al}. \cite{7342961} propose a framework for device-to-device (D2D) coded caching. During the placement phase, similar to coded caching \cite{6620392},  the server fills the users’ caches before the users request any content. During the delivery phase, when each user makes a request for a file, the server is inactive and it is up to the users to transmit signals among themselves so that each user can decode its requested file based on the transmitted signals of the other users and its local cache content. For the centralized D2D coded caching problem, \cite{7342961} used the caching strategy of \cite[Algorithm 1]{6620392}, which is uncoded, and devised a novel delivery scheme fit for the D2D scenario. Furthermore, a widely recognized D2D caching converse was proposed in \cite{7342961}, and it has been shown that the proposed D2D caching and delivery scheme is order optimal within a constant factor when the memory size is large. However, the optimal caching and delivery scheme and the corresponding optimal rate-memory tradeoff for the centralized D2D coded caching problem remain open. 

The optimal caching and delivery scheme for the centralized D2D coded caching problem was characterized in \cite{8830435} under the assumption that the cache placement and delivery are constrained to be \emph{uncoded} and \emph{one-shot}, respectively. One-shot delivery schemes satisfy the condition that each user can decode any bit of its requested file from its own cache and the transmitted signal from at most one user. 
It has been shown in \cite{8830435} that one-shot delivery schemes are optimal within a factor of $2$ under the constraint of uncoded cache placement, and optimal within a factor of $4$ without the constraint of uncoded cache placement.


In addition to \cite{7342961} and \cite{8830435}, there are many other researches for the D2D coded caching problem, such as distinct cache sizes \cite{8977539}, private caching  \cite{9770800},  private caching with a trusted server \cite{9149038,9174279}, secure coded caching \cite{8832245}, secure delivery \cite{7247223}, finite file packetizations \cite{9133151}, wireless multi-hop D2D networks \cite{8007070,9298816}, partially cooperative D2D communication networks \cite{8255121}\cite{9834579}, placement delivery array (PDA)-based design \cite{8620232} and so on. 
Among these papers, to the best of our knowledge, only \cite{7342961} studies the fundamental limits of centralized D2D coded caching allowing coded placement since the nature of the problem is complex and therefore difficult to solve. 
More specifically, the converse result given in \cite{7342961} characterizes the performance of schemes allowing coded cache placement, while the achievability scheme proposed in \cite{7342961} employs uncoded cache placement. Furthermore, \cite{7342961} shows that the proposed scheme is order optimal within a constant factor when the number of users is less than the number of files and the memory size of the users is not very small.

The D2D coded caching problems discussed above assume that all users will make a  file request at the beginning of the delivery phase. However, in practice, this may not always be true. It is possible that some users do not request any files, or they request a file at a much later time than the other users, for example, after the delivery phase of the other uses have been completed. Hence, investigating the D2D coded caching problem where the number of requesters is less than the number of total users is of practical interest. Note that this is different from decentralized D2D coded caching \cite{7342961,8832245,9133151,8904142}, and coded caching with offline users \cite{ArXivPaper,9834866,8830435}, because the users who do not request are still present in the delivery phase and will help with the transmission. 

\subsection{Main Contributions}
In this paper, noting that the number of requesters may be smaller than the number of users, we propose a new D2D coded caching model, called \emph{the 3-user D2D coded caching with two random requesters and one sender (2RR1S)}, where during the delivery phase, any two out of the three users will request a file while the user that is not requesting any files is the designated sender. 
%
The main contributions of the paper can be summarized as follows.

\noindent
1) 
For the \emph{the 3-user D2D coded caching with 2RR1S}, we characterize the fundamental performance limits, i.e., the optimal rate-memory tradeoff, for any number of files $N$. We also find the optimal caching and delivery scheme, called the \emph{2RR1S scheme}, which requires coded cache placement.\\
\noindent
2) To illustrate the usefulness of the proposed model and the 2RR1S scheme, we adapt the 2RR1S scheme to three D2D coded caching scenarios:\\
a) The first scenario is  the 3-user D2D coded caching problem \cite{7342961}. 
 Using the 2RR1S scheme as a baseline scheme, we propose a new coding and delivery scheme, named \emph{the rotated 2RR1S scheme}, which employs coded cache placement. 
By characterizing the optimal rate-memory tradeoff for the 3-user D2D coded caching when $N=2$, which was previously unknown, we show that the rotated 2RR1S scheme proposed is in fact optimal for the 3-user D2D coded caching problem when $N=2$ and the cache size is medium. The benefit comes from \emph{coded} cache placement which is missing from existing D2D coded caching schemes \cite{7342961,8830435}.\\ 
\noindent
b) The second scenario is the case where at the beginning of the delivery phase, each user makes a file request randomly and independently with the same probability $p$, and all three users participate in sending the delivery signals to satisfy the users' requests.  We call this problem the \emph{request-random D2D coded caching}. 
We adapt the 2RR1S scheme to this setting and show that its performance is better than existing D2D coded caching schemes for medium to large cache size.\\
\noindent
c) The third scenario extends the 3-user D2D coded caching with 2RR1S to the \emph{K-user D2D coded caching with $K-s$ random requesters and $s$ senders} problem. We obtain an achievability result for this problem, which is inspired by the 2RR1S scheme and employs coded cache placement.

\section{System Model}\label{System Model}
We introduce a new model in this paper, called the \emph{3-user D2D coded caching with two random requesters and one sender (2RR1S)}. 
There is a server connected to a database of $N$ independent files, $W_1,...,W_N$, and each file consists of $F$ bits, i.e.,
\begin{align}
	H(W_1)&=H(W_2)=\cdots=H(W_N)=F, \nonumber\\
	H(W_1,W_2,\cdots, W_N)&=H(W_1)+H(W_2)+\cdots+H(W_N). \nonumber
\end{align}
There are $K=3$ users in the system, each with a cache of size $MF$ bits, $M\leq N$. The system operates in two phases. In the placement phase, each user's cache is filled with a function of the $N$ files, where we denote the content in the cache of User $k$ as $Z_k$, $k=1,2,3$. 
%
In the delivery phase, any 2 out of the 3 users will make a file request, and the file requests are known to all $3$ users. The user who does not make the file request will send a signal $X^D$, where $D$ denotes the request triple. The signal $X^D$ is received correctly by the two users with file requests, and it is required that each of these two users can decode its requested file using the signal received and its own cache content. We say that the request vector $D=(0,d_2,d_3)$ when Users 2 and 3 request Files $W_{d_2}$ and $W_{d_3}$, respectively, and User 1 does not request anything and is the designated sender. Similarly, the request vector $D$ can take the values of $D=(d_1,0,d_3)$ and $D=(d_1,d_2,0)$, $d_1,d_2,d_3 \in [N]$. Note that we use $[N]$ to represent the set $\{1,2,\cdots,N\}$ and $W_{[N]}$ to represent the set $\{W_1,W_2,\cdots,W_N\}$.

More specifically, a caching and delivery scheme for this system model consists of

\noindent
1) three caching functions
	\begin{equation*}
	\varphi_k : [2^F]^N\to[2^{M F}], k=1,2,3,
	\end{equation*}
	which maps the $N$ files into the cached contents of the users, denoted by $Z_k=\varphi_k(W_1,\cdots,W_N)$, $k=1,2,3$. Since the cached contents are deterministic functions of the files, we have
\begin{align}
H(Z_1, Z_2, Z_3|W_{[N]})=0.\label{1711}
\end{align}	
	\noindent
	2) $3N^2$ encoding functions
	\begin{equation*}
	\phi^D: [2^{MF}]\to[2^{R^D(M) F}], 
	\end{equation*}
	i.e., the encoding function $\phi^D$ denotes the mapping from the cached content of the sender to the signal sent by the sender, and this mapping is a function of the file requests of the other two users. We use $X^{(0,d_2,d_3)}$ to denote the signal sent by User 1, when Users 2 and 3 are requesting files $W_{d_2}$ and $W_{d_3}$, respectively, i.e., $X^{(0,d_2,d_3)}=\phi^{(0,d_2,d_3)}(Z_1)$. Similarly, we define $X^{(d_1,0,d_3)}$ and $X^{(d_1,d_2,0)}$, $d_1,d_2,d_3 \in [N]$. The signal transmitted by the sender for request vector $D$ consists of $R^D(M)F$ bits, where $M$ is the cache size of the three users. Thus, we have
	\begin{align}
H(X^{(0,d_2,d_3)}|Z_1)=0, \quad H(X^{(d_1,0,d_3)}|Z_2)=0, \quad H(X^{(d_1,d_2,0)}|Z_3)=0. \label{eq:0}
\end{align}

\noindent	
3) $6N^2$ decoding functions
	\begin{equation*}
	\psi^D_k: [2^{MF}]\times[2^{R^D(M)F}]\to[2^F], k\in\{i|\text{ the $i$-th element of } D \neq 0\}, 
	\end{equation*}
	which is the decoding function used at User $k$, when the request vector is $D$.

It is required that the caching and delivery scheme enables correct decoding at the users requesting files, i.e., 
\begin{equation}
\begin{aligned}\label{eq:1}
H(W_{d_2}|Z_2, X^{(0,d_2,d_3)})=0, \qquad &H(W_{d_3}|Z_3, X^{(0,d_2,d_3)})=0,\qquad
H(W_{d_1}|Z_1, X^{(d_1,0,d_3)})=0, \\
 H(W_{d_3}|Z_3, X^{(d_1,0,d_3)})=0,\qquad 
&H(W_{d_1}|Z_1, X^{(d_1,d_2,0)})=0, \qquad H(W_{d_2}|Z_2, X^{(d_1,d_2,0)})=0,
\end{aligned}
\end{equation}
which is called the decodability constraint. 
We see from (\ref{eq:1}) that for the decodability constraint to be satisfied, the contents of the cache of any two users must be able to fully recover all $N$ messages, i.e.,
\begin{equation}\label{eq:2}
H(W_{[N]}|Z_\mathcal{I})=0,\qquad \mathcal{I}=\{1,2\}, \{2,3\}, \{3,1\}.
\end{equation}
which means that we must have $2M\geq N$.

Since in the delivery phase, any two of the three users may make file requests and the delivery needs to be done by the third user through a common link between itself and the two users, we call this problem \emph{the 3-user D2D coded caching with two random requesters and one sender (2RR1S)}. 
The schematic diagram of the system model is shown in Fig.~\ref{fig:1}.
\begin{figure}[htbp]
	\centerline{\includegraphics[width=8cm]{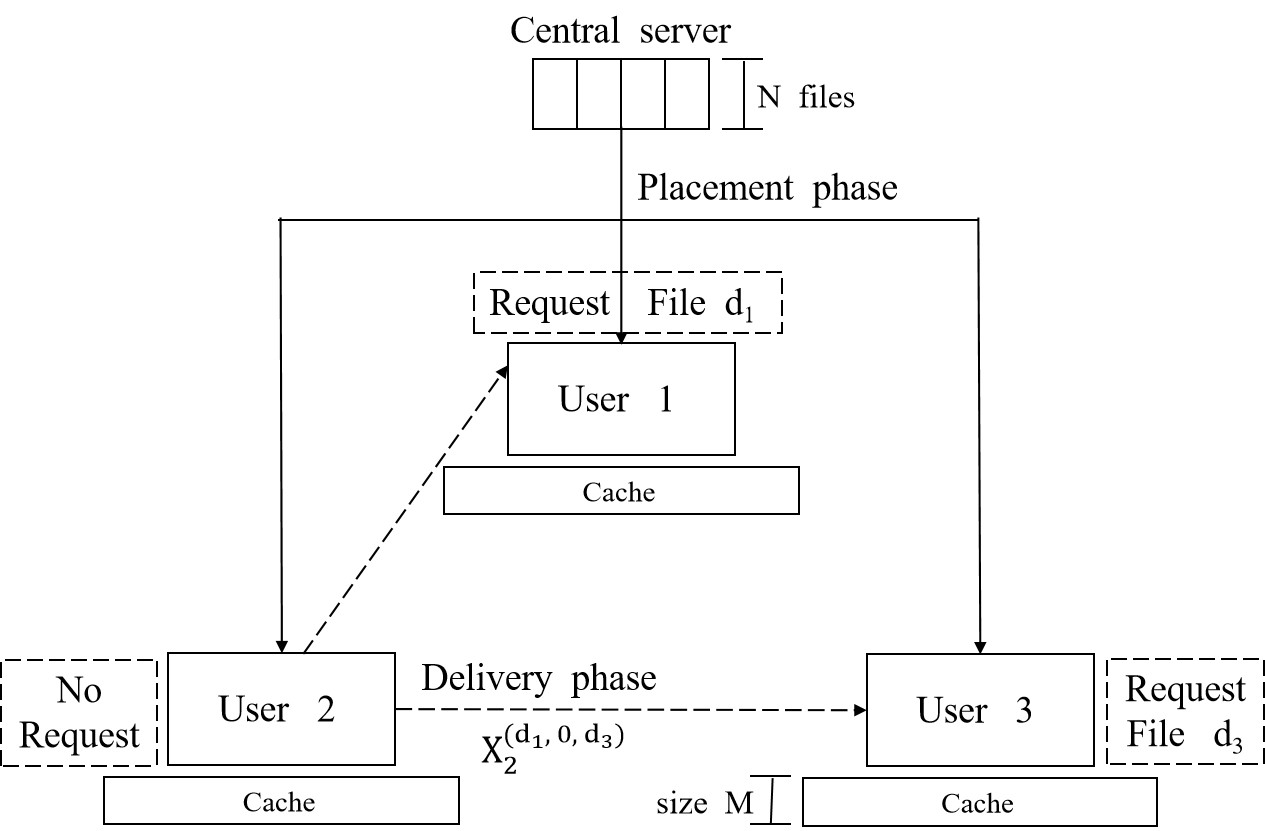}}
	\caption{System model for the 3-user D2D coded caching with two random requesters and one sender. In this realization, User 2 does not request and therefore is the designated sender. Solid and dotted lines indicate the placement and delivery phases, respectively.}
	\label{fig:1}
\end{figure}

For a given caching and delivery scheme that satisfies the decodability constraint (\ref{eq:1}), the performance metric of interest is the worst-case delivery rate, i.e., $R(M) \triangleq \max_D R^D(M)$, where $R^D(M) F$ is the number of symbols transmitted to satisfy  demand $D$. The minimum achievable worst-case rate is given by
\begin{equation*}
R^*(M) \stackrel{\triangle}{=}\inf R(M),
\end{equation*}
where the infimum is taken over all possible caching and delivery schemes that satisfy (\ref{1711}), (\ref{eq:0}) and (\ref{eq:1}). 

To simplify notation, we drop the normalization measure $F$ in the rest of the paper, where the value of $H(W_i)$ is normalized as ``1'', $\forall i$ \cite{8733056}.

\subsection{Symmetric schemes} \label{220919a}
We observe that the characteristic of symmetry \cite[Section 3]{e20080603} applies to 
the problem of the 3-user D2D coded caching with 2RR1S. 
%
More specifically, let $\bar{\pi}(\cdot)$ be a permutation function on the user index set $\{1,2,3\}$, and denote its inverse function as $\bar{\pi}^{-1}(\cdot)$. Further let $\mathcal{Z} \subseteq \{Z_1,Z_2,Z_3\}$, $\mathcal{X} \subseteq \{X^{(0,d_2,d_3)},X^{(d_1,0,d_3)},X^{(d_1,d_2,0)}| d_1, d_2,d_3 \in [N]\}$. The mapping $\bar{\pi}(\mathcal{Z})$ denotes $\{Z_{\bar{\pi}(k)}| Z_k\in\mathcal{Z}\}$ and the mapping $\bar{\pi}(\mathcal{X})$ denotes $\left\{X^{\left(d_{\bar{\pi}^{-1}(1)},d_{\bar{\pi}^{-1}(2)},d_{\bar{\pi}^{-1}(3)} \right)}| \break X^{(d_1,d_2,d_3)}\in\mathcal{X}\right\}$.
\emph{User-index-symmetric schemes} \cite[Section 3]{e20080603} are defined as follows. 
\begin{Def}\label{Definition 3}
	\normalfont
	A caching and delivery scheme is called \emph{user-index-symmetric} if for any permutation function $\bar{\pi}(\cdot)$, any subset of files $\mathcal{W}$, any subset of caches $\mathcal{Z}$, and any subset of transmitted signals $\mathcal{X}$, we have the following relation:
	\begin{equation}
		H(\mathcal{W},\mathcal{Z},\mathcal{X})=H(\mathcal{W},\bar{\pi}(\mathcal{Z}),\bar{\pi}(\mathcal{X})). \nonumber
	\end{equation}
\end{Def}

For example, consider the permutation function $\bar{\pi}(1)=2,\bar{\pi}(2)=3,\bar{\pi}(3)=1$. For a user-index-symmetric scheme for the 3-user D2D coded caching with 2RR1S, the entropy $H(W_1,Z_1,\break X^{(1,0,2)})$ under the permutation $\bar{\pi}$ is equal to $H(W_1,Z_2,X^{(2,1,0)})$.

Similar to the definition of user-index-symmetric schemes, \emph{file-index-symmetric schemes} \cite[Section 3]{e20080603} may be defined. More specifically, let $\hat{\pi}(\cdot)$ be a permutation function on the file index set $\{1,2,\cdots,N\}$, and further let $\hat{\pi}(0)=0$. Let $\mathcal{W} \subseteq \{W_1,W_2,\cdots, W_N\}$, and by representing the mapping $\hat{\pi}(\mathcal{W})$ as $\{W_{\hat{\pi}(n)}| W_n\in\mathcal{W}\}$ and the mapping $\hat{\pi}(\mathcal{X})$ as 
$ \{X^{(\hat{\pi}(d_1),\hat{\pi}(d_2),\hat{\pi}(d_3))}| X^{(d_1,d_2,d_3)}\in\mathcal{X}\}$, \emph{File-index-symmetric schemes}  \cite[Section 3]{e20080603} are defined as follows. 
\begin{Def}\label{Definition 4}
	\normalfont
	A caching and delivery scheme is called \emph{file-index-symmetric} if for any permutation function $\hat{\pi}(\cdot)$, any subset of files $\mathcal{W}$, any subset of caches $\mathcal{Z}$, and any subset of transmitted signals $\mathcal{X}$, we have the following relation:
	\begin{equation}
		H(\mathcal{W},\mathcal{Z},\mathcal{X})=H(\hat{\pi}(\mathcal{W}),\mathcal{Z},\hat{\pi}(\mathcal{X})). \nonumber
	\end{equation}
\end{Def}

For example, for the 3-user D2D coded caching with 2RR1S, if User 2 does not request, the permutation function $\hat{\pi}(0)=0,\hat{\pi}(1)=2,\hat{\pi}(2)=3,\hat{\pi}(3)=1$, will map $W_1$ to $\hat{\pi}(W_1)=W_2$, but map $X^{(1,0,2)}$ to $X^{(2,0,3)}$. For such a file-index-symmetric scheme, the entropy $H(W_1,Z_1,X^{(1,0,2)})$ under the permutation is equal to $H(W_2,Z_1,X^{(2,0,3)})$.

The characteristics of symmetry \cite[Section 3]{e20080603} applies to the 3-user D2D coded caching with 2RR1S problem. More specifically, we have the following lemma for the 3-user D2D coded caching with 2RR1S, whose proof is similar to that of \cite[Proposition 3.1]{e20080603}.
\begin{Lem}\label{Lemma0}
	For the 3-user coded caching problem with 2RR1S, for any caching and delivery scheme, there exists a caching and delivery scheme which is both user-index-symmetric and file-index-symmetric with an equal or smaller worst-case delivery rate.
\end{Lem}
\begin{proof}
Consider a base code, which may not be user-index-symmetric or file-index-symmetric. For demand $D$, the delivery rate is $R^{D}(M)$. Its worst-case delivery rate $R(M)=\max_{D}R^{D}(M)$. We form a new code as follows. Split each file into $N!K!$ segments, each having the same size. For each file segment, use the base code with a different user-index and file-index joint permutation to achieve $N!K!$ relevant codes, whose rate is $R^{D}_{\bar{\pi}\hat{\pi}}(M)$ for demand $D$. Then, form a new code by the space sharing of all the $N!K!$ codes. Due to the symmetry of construction, the new code is indeed user-index-symmetric and file-index symmetric. Furthermore, the delivery rate of the new code for demand $D$ is 
\begin{align}
	\bar{R}^D(M)=\sum_{ \pi, \hat{\pi}} \frac{1}{N! K!} R^{D}_{\bar{\pi}\hat{\pi}}(M) \overset{(a)}{\leq} \sum_{\pi, \hat{\pi}} \frac{1}{N! K!} R(M)=R(M) \label{0424a}
\end{align}
where $(a)$ follows from the fact that $R(M)$ is the worst-case delivery rate over all possible demand $D$. From (\ref{0424a}), taking the maximum over $D$ on both sides, we obtain that the worst-case delivery rate of the new code $\bar{R}(M)=\max_D \bar{R}^D(M)$ is no larger than the worst-case delivery rate of the base code, i.e., $R(M)$, which proves Lemma~\ref{Lemma0}.
\end{proof}

This observation of Lemma~\ref{Lemma0} can be used to simplify the proof of the converse, 
where it is sufficient to consider only caching and delivery schemes that satisfy both user-index symmetry and file-index symmetry.

\section{Main Results on the 3-user D2D coded caching with 2RR1S} \label{SecNan01}
In this paper, we find the optimal rate-memory tradeoff for the proposed 3-user D2D coded caching with 2RR1S for any number of files. It turns out that the rate-memory tradeoff satisfies a uniform formula in the case of more than 4 files, and takes on distinct formulas in the case of 2 files and 3 files. More specifically, we have the following theorem when $M \geq \frac{N}{2}$. Note that when $M< \frac{N}{2}$, the problem is infeasible, i.e., there exists no caching and delivery scheme that can satisfy the decodability constraint (\ref{eq:1}). 
\begin{Theo} \label{MainResult}
For the problem of the 3-user D2D coded caching with 2RR1S where $M \geq \frac{N}{2}$,

\noindent
(1) For $N\geq4$, the worst-case delivery rate $R(M)$
 must satisfy
\begin{equation}
4M+NR(M)\geq3N, \qquad M+NR(M)\geq N, \label{eq:8} 
\end{equation}
where the corner points are $(M,R(M))=(\frac{1}{2}N,1)$, $(\frac{2}{3}N,\frac{1}{3})$, $(N,0)$.
Conversely, there exist caching and delivery schemes for any nonnegative $R(M)$ satisfying (\ref{eq:8}).

\noindent
(2) For $N=2$,  the worst-case delivery rate $R(M)$
must satisfy
\begin{equation}
18M+8R(M)\geq25, \qquad 3M+3R(M)\geq5, \qquad M+2R(M)\geq2\label{eq:6},
\end{equation}
where the corner points are $(M,R(M))=(1,\frac{7}{8})$, $(\frac{7}{6},\frac{1}{2})$, $(\frac{4}{3},\frac{1}{3})$, $(2,0)$. Conversely, there exist caching and delivery schemes for any nonnegative $R(M)$ satisfying (\ref{eq:6}).

\noindent
(3) For $N=3$,  the worst-case delivery rate $R(M)$
must satisfy
\begin{equation}
6M+4R(M)\geq13, \qquad 3M+3R(M)\geq7, \qquad M+3R(M)\geq3\label{eq:7},
\end{equation}
where the corner points are $(M,R(M))=(\frac{3}{2},1)$, $(\frac{11}{6},\frac{1}{2})$, $(2,\frac{1}{3})$, $(3,0)$. Conversely, there exist caching and delivery schemes for any nonnegative $R(M)$ satisfying (\ref{eq:7}).

\end{Theo}
The proof of Theorem \ref{MainResult} is given in Section \ref{SecProof}. 
We make the following remarks regarding of the result of Theorem \ref{MainResult}, including comparisons with existing work. 
\begin{Remark}\label{Remark 0}
The performance of the 3-user D2D coded caching with 2RR1S is upper bounded by the optimal performance of the original 2-user coded caching problem where the sender is the server. 
This is because the server knows everything and is more capable than any of the D2D sender nodes. The optimal rate-memory tradeoff for the original coded caching problem with 2 users and $N$ files was found in \cite{e20080603}, and the converse result of $M+NR(M)\geq N$ is proved. From Theorem \ref{MainResult}, we see that when the memory is large enough, i.e., $M\in[\frac{2}{3}N,N]$, $M+NR(M) =N$ is achievable for the 3-user D2D coded caching with 2RR1S, which means that when the memory is large, the random D2D sender node is as capable as the all-knowing server.
\end{Remark}

\begin{Remark}\label{Remark 5}
\normalfont
We find that the proposed optimal scheme employs \emph{coded} cache placement when $M\in[\frac{1}{2}N,\frac{2}{3}N)$, while uncoded  cache placement is sufficient when $M\in[\frac{2}{3}N,N]$. This observation is similar to the result of the traditional coded caching problem in \cite{e20080603} where uncoded placement is sufficient, i.e., optimal, when $M\in[\frac{K-1}{K}N,N]$.
\end{Remark}

\begin{Remark}\label{Remark 4}
\normalfont
The corner point of $\left(\frac{1}{2}N,1\right)$ for $N \geq 2$, the corner points of $\left(1,\frac{7}{8}\right)$, $\left( \frac{7}{6}, \frac{1}{2}\right)$ for $N=2$, and the corner point $\left( \frac{11}{6}, \frac{1}{2}\right)$ for $N=3$ all employ \emph{coded} cache placement. More specifically, to deal with the fact that the identity of the transmitter is unknown at the time of cache placement, the cache content of the users are MDS coded across the three users. Furthermore, in the case of the corner points $\left( \frac{7}{6}, \frac{1}{2}\right)$ for $N=2$ and $\left( \frac{11}{6}, \frac{1}{2}\right)$ for $N=3$, transmitter preprocessing is required, where the designated sender needs to do some computation to obtain the transmitted signal. 
 We also show in the achievability proof that the two corner points $\left( \frac{7}{6}, \frac{1}{2}\right)$ for $N=2$ and $\left( \frac{11}{6}, \frac{1}{2}\right)$ for $N=3$ belong to the more general set of achievable corner points $\left(\frac{4N-1}{6},\frac{1}{2}\right)$ for $N \geq 2$. These corner points are optimal for $N=2,3$, but for $N \geq 4$, they are sub-optimal. 

\end{Remark}

\begin{Remark}\label{Remark 2}
\normalfont
In the problem considered, we have three users, i.e., $K=3$. From Theorem~1, we see that the number of corner points is different for the case of $N>K$ and the case of $N\leq K$. This is because the corner point $(\frac{4N-1}{6},\frac{1}{2})$ is below the converse line $4M+NR(M)\geq 3N$ only when $N\leq K$, which means that $4$ corner points exist when $N\leq K$ and $3$ corner points  exist when $N>K$. Note that in the traditional coded caching problem, e.g., \cite{8226776}, and other D2D coded caching problems where a tight converse exists, e.g.,  \cite{8830435}, the number of corner points are different for the case of $N\geq K$ and the case of $N<K$. 
\end{Remark}

\section{Proof of Theorem \ref{MainResult}} \label{SecProof}
In this section, we prove the converse and achievability for Theorem \ref{MainResult}. The proof is different for $N \geq 4$, $N=3$ and $N=2$, where $N$ is the number of files.  
\subsection{Achievability}\label{Ach}
\subsubsection{ $N \geq 4$}\label{AchN>4}
We show that the three corner points $(\frac{1}{2}N,1)$, $(\frac{2}{3}N,\frac{1}{3})$ and $(N,0)$ are achievable as long as $N \geq 2$. It will be shown via the converse proof that the three corner points $(\frac{1}{2}N,1)$, $(\frac{2}{3}N,\frac{1}{3})$ and $(N,0)$ are optimal only when $N \geq 4$ is satisfied. 

The corner point of $(N,0)$ is trivial as all D2D nodes have enough cache to store all messages and therefore, the delivery rate is zero. As for the corner point of $(\frac{1}{2}N,1)$, its achievability scheme is as follows: split all files into two subfiles of equal sizes, denoted as $W_n=(W_{n,1},W_{n,2})^N_{n=1}$. In the cache placement phase, the cache content of the three users are given as
\begin{align*}
Z_1=(W_{n,1}\oplus W_{n,2})^N_{n=1}, 
Z_2=(W_{n,1})^N_{n=1}, 
Z_3=(W_{n,2})^N_{n=1}. 
\end{align*}
In the delivery phase, we have
\begin{align*}
&X^{(0,d_2,d_3)}=\{ W_{d_2,1} \oplus W_{d_2,2}, W_{d_3,1} \oplus W_{d_3,2}\}, \nonumber\\
&X^{(d_1,0,d_3)}=\{W_{d_1,1}, W_{d_3,1}\},
X^{(d_1,d_2,0)}=\{W_{d_1,2}, W_{d_2,2}\}, 
\end{align*}
and it is easy to check that each user's demand can be correctly decoded. Thus, the delivery rate of $R(M)=1$ is achieved for cache size $M=\frac{1}{2}N$. 
As can be seen, coded caching is necessary to achieve the corner point of $(\frac{1}{2}N,1)$. 

Lastly, we provide the achievability scheme for the corner point of $\left(\frac{2}{3}N, \frac{1}{3} \right)$. The caching scheme is the same as that of the Maddah-Ali Niesen (MAN) uncoded symmetric placement in \cite[Algorithm 1]{6620392}, more specifically, all files are split into three subfiles of equal sizes, denoted as $W_n=(W_{n,\left\{1,2\right\}},W_{n,\left\{1,3\right\}},W_{n,\left\{2,3\right\}})^N_{n=1}$, and the cache placement is
\begin{align*}
Z_1=(W_{n,\left\{1,2\right\}}, W_{n,\left\{1,3\right\}})^N_{n=1},
Z_2=(W_{n,\left\{1,2\right\}}, W_{n,\left\{2,3\right\}})^N_{n=1},
Z_3=(W_{n,\left\{1,3\right\}}, W_{n,\left\{2,3\right\}})^N_{n=1}.
\end{align*}
In the delivery phase, we have
\begin{align*}
&X^{(0,d_2,d_3)}=\{ W_{d_2,\{1,3\}} \oplus W_{d_3,\{1,2\}}\}, 
X^{(d_1,0,d_3)}=\{W_{d_1,\{2,3\}} \oplus W_{d_3,\{1,2\}}\}, \nonumber\\
&X^{(d_1,d_2,0)}=\{W_{d_1,\{2,3\}} \oplus W_{d_2,\{1,3\}}\}, 
\end{align*}
and it is easy to check that each user's demand can be correctly decoded. Thus, the delivery rate of $R(M)=\frac{1}{3}$ is achieved for cache size $M=\frac{2}{3}N$. 
To achieve the corner point of $\left(\frac{2}{3}N, \frac{1}{3} \right)$, coded cache placement is not necessary, and the MAN symmetric uncoded placement scheme \cite{6620392} is used. 
Finally, memory sharing between the corner points $(\frac{1}{2}N,1)$, $(\frac{2}{3}N,\frac{1}{3})$ and $(N,0)$ proves that (\ref{eq:8}) in Theorem \ref{MainResult} is achievable.

\subsubsection{$N=2$}\label{AchN=2}
We will prove that the four corner points $(1,\frac{7}{8})$, $(\frac{7}{6},\frac{1}{2})$, $(\frac{4}{3},\frac{1}{3})$, $(2,0)$ are achievable. First, note the fact that the corner points $(\frac{4}{3},\frac{1}{3})$, $(2,0)$ are achievable has been proved in Section \ref{AchN>4} which works for $N= 2$. So in the following, we will prove that the remaining two points are achievable. 

To achieve the corner point of $(1,\frac{7}{8})$, 
we split both files into $8$ subfiles, i.e., $W_1=(A_n)^8_{n=1}$ and $W_2=(B_n)^8_{n=1}$. Coded cache placement is employed as shown in Table~\ref{table 1}, where the cache size is indeed $1$.
%
%
\begin{table}[htb]
\centering
\caption{ Cache placement for achieving the corner point $(1,\frac{7}{8})$ when $N=2$}
\label{table 1}
\begin{tabular}{|c|c c c c c c c c|}
  \hline
 $Z_1$: &$A_1\oplus B_2$&$A_2\oplus B_1$&$B_4$&$A_4$&$A_5$&$B_5$&$A_7\oplus A_8$&$B_7\oplus B_8$\\
  \hline
  $Z_2$: &$A_1$&$B_1$&$A_3\oplus B_4$&$A_4\oplus B_3$&$B_6$&$A_6$&$A_7$&$B_7$\\
  \hline
$Z_3$: &$B_2$&$A_2$&$A_3$&$B_3$&$A_5\oplus B_6$&$A_6\oplus B_5$&$A_8$&$B_8$\\
  \hline
\end{tabular}
\end{table}

For the delivery phase, the transmitted signal depends on the request vector as: 
\begin{align}
	X^{(0,d_2,d_3)}&=\{W_{d_2,2}\oplus W_{d_3,1}, A_4, B_4, A_5, B_5, A_7\oplus A_8, B_7\oplus B_8\}, \quad d_2\neq d_3, 
	\nonumber\\
	X^{(0,d_2,d_3)}&=\{W_{d_2,7}\oplus W_{d_2,8}, A_4, B_4, A_5, B_5, A_1\oplus B_2, A_2\oplus B_1\}, \quad d_2= d_3, 
	\label{0421a}\\
	X^{(d_1,0,d_3)}&=\{W_{d_1,3}\oplus W_{d_3,4}, A_1, B_1, A_6, B_6, A_7, B_7\}, \quad d_1 \neq d_3, 
	\nonumber\\
	X^{(d_1,0,d_3)}&=\{W_{d_1,7}, A_1, B_1, A_6, B_6, A_3\oplus B_4, A_4\oplus B_3\}, d_1=d_3, 
	\nonumber\\
	X^{(d_1,d_2,0)}&=\{W_{d_1,6}\oplus W_{d_2,5}, A_2, B_2, A_3, B_3, A_8, B_8\}, \quad d_1 \neq d_2,
	\nonumber\\
	X^{(d_1,d_2,0)}&=\{W_{d_1,8}, A_2, B_2, A_3, B_3, A_5\oplus B_6, A_6\oplus B_5\}, \quad d_1=d_2,
	\nonumber
\end{align}
where $d_1, d_2, d_3\in\{1,2\}$. As can be seen, when the two random requesters request different files, both the 7-th and the 8-th elements of the designated sender's cache are transmitted, but one of the first six elements do not need to be transmitted, resulting in a delivery rate of $\frac{7}{8}$. On the other hand, when the two random requesters request the same file, the first six elements of the designated sender's cache are transmitted, and only one of the 7-th or 8-th element is transmitted, i.e., if both random requesters request file $W_1$, then the 7-th element of the designated sender's cache is transmitted, and if both random requesters request file $W_2$, then the 8-th element of the designated sender's cache is transmitted. This again results in a delivery rate of $\frac{7}{8}$.
It is easy to check that the decoding constraint is satisfied.
Note that even though the caching and delivery scheme looks asymmetrical in user index in terms of expressions, it is in fact user-index symmetrical in terms of entropy.


To achieve the corner point of $(\frac{7}{6},\frac{1}{2})$, 
we split both files into $6$ subfiles, i.e., $W_1=(A_n)^6_{n=1}$ and $W_2=(B_n)^6_{n=1}$. The caching scheme is given in Table~\ref{table 2}, where the cache size is indeed $\frac{7}{6}$.

\begin{table}[!htbp]
\centering
\caption{Cache placement for achieving the corner point $(\frac{7}{6},\frac{1}{2})$ when $N=2$}
\label{table 2}
\begin{tabular}{|c|c c c c c c c|}
  \hline
  $Z_1$: &$A_1\oplus A_2$&$B_1\oplus B_2$&$A_4$&$B_4$&$A_5$&$B_5$&$A_2\oplus B_1$\\
  \hline
  $Z_2$: &$A_1$&$B_1$&$A_3\oplus A_4$&$B_3\oplus B_4$&$A_6$&$B_6$&$A_4\oplus B_3$\\
  \hline
  $Z_3$: &$A_2$&$B_2$&$A_3$&$B_3$&$A_5\oplus A_6$&$B_5\oplus B_6$&$A_6\oplus B_5$\\
  \hline
\end{tabular}
\end{table}

For the delivery phase, the transmitted signal depends on the request vector as
\begin{align}
	X^{(0,d_2,d_3)}&=\{W_{d_2,2}\oplus W_{d_3,1}, W_{d_3,4}, W_{d_2,5}\},
	X^{(d_1,0,d_3)}=\{W_{d_1,3}\oplus W_{d_3,4}, W_{d_3,1}, W_{d_1,6}\},
	\label{delivery1/2}\\
	X^{(d_1,d_2,0)}&=\{W_{d_1,6}\oplus W_{d_2,5}, W_{d_2,2}, W_{d_1,3}\},
	\label{delivery1/22}
\end{align}
 where $d_1, d_2, d_3\in\{1,2\}$. 
 Note that the first six columns are MDS coded across the three users, so that any two of them can recover the entire segment, for example, the first column can recover the segment $(A_1, A_2)$ using the cache of any two users, and the second column can recover the segment $(B_1,B_2)$ using the cache of any two users. The last column is a coded version that enables User 1 to transmit any pairwise linear combination of $(A_1,A_2,B_1,B_2)$, User 2 to transmit any pairwise linear combination of $(A_3,A_4,B_3,B_4)$, and User 3 to transmit any pairwise linear combination of $(A_5,A_6,B_5,B_6)$. This offers flexibility in the delivery signal of the designated sender based on the demand of the other two users, i.e., 
  transmitter preprocessing is needed. For example, when the request vector is $D=(0,2,1)$, i.e., User 1 does not request and is the designated sender, Users 2 and 3 requests Files $W_2$ and $W_1$, respectively, the transmitted signal is $(B_2 \oplus A_1, A_4, B_5)$. As can be seen, $B_2 \oplus A_1$ is not directly stored in the cache of User 1. Preprocessing at User 1 as
$
 B_2 \oplus A_1=(A_1\oplus A_2)\oplus(B_1\oplus B_2)\oplus(A_2\oplus B_1),
$
is needed before $B_2 \oplus A_1$ is transmitted. More generally, when User $k$ needs to transmit $A_{2k-1}\oplus B_{2k}$ which is not cached, User $k$ does the following preprocessing:
$
	A_{2k-1}\oplus B_{2k}=(A_{2k-1}\oplus A_{2k})\oplus(B_{2k-1}\oplus B_{2k})\oplus(A_{2k}\oplus B_{2k-1}),  k=1,2,3. 
$
The fact that the decoding constraint is satisfied can be checked. We mention here that sometimes several modulo-sums need to be computed to decode, rather than just one modulo-sum. For example, in the case where the request vector is $D=(0,2,1)$ as discussed above, upon receiving $(B_2 \oplus A_1, A_4, B_5)$, which is the transmitted signal of User 1, User 2 decodes $B_3$ and $B_4$ by first computing $(A_3 \oplus A_4) \oplus A_4$ to obtain $A_3$, and then decode $B_4$ as $(A_3 \oplus B_4) \oplus A_3$, and finally, decode $B_3$ as $(B_3 \oplus B_4) \oplus B_4$. Note that the transmission rate of the proposed scheme is $\frac{3}{6}=\frac{1}{2}$. 
%
%
%
%
%
Memory sharing between the corner points $(1,\frac{7}{8})$, $(\frac{7}{6},\frac{1}{2})$, $(\frac{4}{3},\frac{1}{3})$ and $(2,0)$ proves that (\ref{eq:6}) in Theorem \ref{MainResult} is achievable.
\subsubsection{$N=3$}\label{AchN=3}
In the case of $N=3$, the achievability of corner points $(\frac{3}{2},1)$, $(2,\frac{1}{3})$ and $(3,0)$ has been proven in Section \ref{AchN>4} which works for $N=3$. Thus, we only need to prove the achievability of the corner point $(\frac{11}{6},\frac{1}{2})$. 

We split the three files $W_1, W_2, W_3$ into 6 subfiles, which can be represented as $W_1=(A_n)^6_{n=1}$, $W_2=(B_n)^6_{n=1}$ and $W_3=(C_n)^6_{n=1}$. Coded cache placement is employed as shown in Table~\ref{table 3}, where the cache size is indeed $\frac{11}{6}$. 

\begin{table}[htb]
\centering
\caption{ Cache placement for achieving the corner point $(\frac{11}{6},\frac{1}{2})$ when $N=3$}
\label{table 3}
\begin{tabular}{|c|c c c c c c|}
  \hline
  $Z_1$: &$A_1\oplus A_2$&$B_1\oplus B_2$&$C_1 \oplus C_2$&$A_4$&$B_4$&$C_4$\\
  \hline
  $Z_2$: &$A_1$&$B_1$&$C_1$&$A_3\oplus A_4$&$B_3\oplus B_4$&$C_3 \oplus C_4$ \\
  \hline
  $Z_3$: &$A_2$&$B_2$&$C_2$&$A_3$&$B_3$&$C_3$ \\
  \hline
  \hline
   $Z_1$: &$A_5$&$B_5$& $C_5$ &$A_2\oplus B_1$  & $B_2 \oplus C_1$ & \\
  \hline
  $Z_2$:  & $A_6$&$B_6$&$C_6$ & $A_4\oplus B_3$& $B_4+C_3$&\\
  \hline
  $Z_3$: &$A_5\oplus A_6$& $B_5\oplus B_6$&$C_5 \oplus C_6$ & $A_6\oplus B_5$ & $B_6 \oplus C_5$&\\
  \hline
\end{tabular}
\end{table}

For the delivery phase, the transmitted signal depends on the request vector as
\begin{align*}
X^{(0,d_2,d_3)}=\{W_{d_2,2}\oplus W_{d_3,1}, W_{d_3,4}, W_{d_2,5}\},\nonumber\\
X^{(d_1,0,d_3)}=\{W_{d_1,3}\oplus W_{d_3,4}, W_{d_3,1}, W_{d_1,6}\},\nonumber\\
X^{(d_1,d_2,0)}=\{W_{d_1,6}\oplus W_{d_2,5}, W_{d_2,2}, W_{d_1,3}\}.
\end{align*}
It can be checked that the proposed scheme has no decoding error and the transmission rate is $\frac{3}{6}=\frac{1}{2}$. 

The scheme above is a generalization of the scheme that achieves the corner point $\left(\frac{7}{6}, \frac{1}{2} \right)$ in $N=2$. 
The first nine columns are MDS coded across the three users, so that any two of them can recover the entire segment. The last two columns are a coded version that enables the designated sender to do transmitter preprocessing and send out the signal needed based on the demands of the other two users. 

More generally, for any $N \geq 2$, the corner point $\left(\frac{4N-1}{6},\frac{1}{2}\right)$ is achievable as follows: split all files $W_1, \cdots, W_N$ into six subfiles, i.e., $W_n=(W_{n,1}, \cdots, W_{n,6})$, $n \in [N]$. The cache placement at the three users are given as
\begin{align*}
	Z_1=\{W_{n,1}\oplus W_{n,2}, W_{n,4}, W_{n,5}, W_{n,2}\oplus W_{n+1,1}\}^{N-1}_{n=1} \bigcup \{W_{N,1}\oplus W_{N,2}, W_{N,4}, W_{N,5} \}, \\
	Z_2=\{W_{n,3}\oplus W_{n,4}, W_{n,1}, W_{n,6}, W_{n,4}\oplus W_{n+1,3}\}^{N-1}_{n=1} \bigcup \{W_{N,3}\oplus W_{N,4}, W_{N,1}, W_{N,6}\},\\
	Z_3=\{W_{n,5}\oplus W_{n,6}, W_{n,2}, W_{n,3}, W_{n,6}\oplus W_{n+1,5}\}^{N-1}_{n=1} \bigcup \{W_{N,5}\oplus W_{N,6}, W_{N,2}, W_{N,3}\}.
\end{align*}
The delivery scheme is given by (\ref{delivery1/2}) and (\ref{delivery1/22}). 
%
Note that for the case of $N \geq 4$, the corner point $\left(\frac{4N-1}{6},\frac{1}{2}\right)$, though achievable, is not optimal, i.e., it lies above the memory-sharing curve of the three achievable corner points $(\frac{1}{2}N,1)$, $(\frac{2}{3}N,\frac{1}{3})$ and $(N,0)$. 
Finally, we conclude that memory sharing between the corner points $(\frac{3}{2},1)$, $(\frac{11}{6},\frac{1}{2})$, $(2,\frac{1}{3})$ and $(3,0)$ proves that (\ref{eq:7}) in Theorem \ref{MainResult} is achievable.
\subsection{Converse}
\subsubsection{$N \geq 4$}\label{ProofN>4}
As mentioned in Remark \ref{Remark 0}, 
the performance of the 3-user D2D coded caching with 2RR1S is upper bounded by the optimal performance of the original 2-user coded caching problem where the sender is the server. Thus, the converse result of the original coded caching problem \cite{6620392}, more specifically, 
\begin{align}
M+NR(M)\geq N, \quad N \geq 2, \label{22070401}
\end{align} is also a converse result for the 3-user D2D coded caching with 2RR1S. Hence, we only need to prove $4M+NR(M)\geq3N$. 

To prove $4M+NR(M)\geq3N$, we follow similar steps as those in the proof of \cite[Lemma 1]{e20080603} and utilize the finding in \cite[Lemma 4]{8733056}, which shows $NH(Z_{1}|W_{1})\geq(N-1)H(Z_{1})$ for any file-index-symmetric schemes. Furthermore, the property given in (\ref{eq:2}) for the model under consideration is exploited in the proof. 

More specifically, as discussed in Section \ref{220919a}, we may, without loss of generality, consider only user-index-symmetric and file-index-symmetric caching and delivery schemes. For any  user-index-symmetric and file-index-symmetric caching and delivery scheme with achievable rate $R(M)$, it must satisfy
\begin{align}
&(N-1)H({X^{(1,0,2)}},Z_{1},W_{1})
\nonumber \\
=&(N-1)[H(Z_{1},W_{1})+H(X^{(1,0,2)}|Z_{1},W_{1})]
\nonumber\\
=&(N-1)H(Z_{1},W_{1})+\sum_{i=2}^NH(X^{(1,0,i)}|Z_{1},W_{1}) \label{Nan01}\\
\geq&(N-1)H(Z_{1},W_{1})+H(X^{(1,0,[2:N])}|Z_{1},W_{1})
\nonumber\\
=&(N-2)H(Z_{1},W_{1})+H(X^{(1,0,[2:N])},Z_{1},W_{1})
\nonumber\\
=&H(Z_{3},W_{1})+H(Z_{2},W_{1})+H(X^{(1,0,[2:N])},Z_{1},W_{1})
+(N-4)H(Z_{1},W_{1})
\label{Nan02}\\
=& H(Z_{3},W_{1})+H(Z_{2},X^{(1,0,[2:N])},W_{1})+H(X^{(1,0,[2:N])},Z_{1},W_{1})
+(N-4)H(Z_{1},W_{1}) \label{Nan03}\\
\geq &H(Z_{3},W_{1})+H(Z_{1},Z_{2},X^{(1,0,[2:N])},W_{1})+H(X^{(1,0,[2:N])},W_{1})
+(N-4)H(Z_{1},W_{1})
\label{Nan04}\\
=&H(Z_{3},W_{1})+H(W_{[N]})+H(X^{(1,0,[2:N])},W_{1})
+(N-4)H(Z_{1},W_{1})
\label{Nan05}\\
\geq &N+H(Z_{3},X^{(1,0,[2:N])},W_{1})+H(W_{1})+(N-4)H(Z_{1},W_{1})
\label{Nan06})\\
=&2N+1+(N-4)H(Z_{1},W_{1}), \label{Nan07}
\end{align}
where $X^{(1,0,[2:N])}$ denotes $\{X^{(1,0,2)},\cdots,X^{(1,0,N)}\}$, (\ref{Nan01}) follows from the property of file-index-symmetric schemes,
(\ref{Nan02}) follows from the property of user-index-symmetric schemes, (\ref{Nan03}) follows from (\ref{eq:0}), (\ref{Nan04}) and (\ref{Nan06}) both follow from the sub-modular property of the entropy function, i.e.,
$
H(Y_\mathcal{A})+H(Y_\mathcal{B}) \geq H(Y_{\mathcal{A} \bigcup \mathcal{B}})+H(Y_{\mathcal{A} \bigcap \mathcal{B}}), 
$
(\ref{Nan05}) follows from (\ref{1711}) and (\ref{eq:2}), and (\ref{Nan07}) follows from the decodability constraint in (\ref{eq:1}), more specifically, $H(W_{[2:N]}|Z_{3},X^{(1,0,[2:N])})=0$, where $W_{[2:N]}$ denotes $\{W_2,\cdots,W_N\}$.
Using the result of (\ref{Nan07}), we have
\begin{align}
(N-1)H({X^{(1,0,2)}|Z_{1},W_{1}})
=&(N-1)H({X^{(1,0,2)},Z_{1},W_{1}})-(N-1)H(Z_{1},W_{1})
\nonumber\\
\geq &2N+1+(N-4)H(Z_{1},W_{1})-(N-1)H(Z_{1},W_{1})
\label{Nan08}\\
=&2(N-1)-3H(Z_{1}|W_{1}), \label{Nan09}
\end{align}
where (\ref{Nan08}) follows from (\ref{Nan07}). Finally, we have
\begin{align}
M+R(M)\geq&H(Z_{1})+H(X^{(1,0,2)})\geq H(Z_{1},X^{(1,0,2)},W_{1})
\nonumber\\
=&H(W_{1})+H(Z_{1}|W_{1})+H({X^{(1,0,2)}|Z_{1},W_{1}})
\nonumber\\
\geq &1+H(Z_{1}|W_{1})+2-\frac{3}{N-1}H(Z_{1}|W_{1})
\label{Nan10}\\
=& 3+\frac{N-4}{N-1}H(Z_1|W_1) \nonumber\\
\geq &3+\frac{N-4}{N}H(Z_{1}), \label{Nan11}
\end{align}
where (\ref{Nan10}) follows from (\ref{Nan09}), (\ref{Nan11}) follows \cite[Lemma 4]{8733056} that shows $NH(Z_{1}|W_{1})\geq(N-1)H(Z_{1})$ for any file-index-symmetric schemes. From (\ref{Nan11}), we have
$
4M+NR(M) \geq 3N, 
$
which completes the proof of (\ref{eq:8}) in Theorem \ref{MainResult}.


\subsubsection{$N=2,3$}
The converse proof for the case of $N=2$ and $N=3$ is given in Appendix \ref{0420a}. The proof benefited greatly from the computer-aided discovery of outer bounds presented in \cite{e20080603}. The computer-aided approach is useful for investigating the fundamental limits of coded caching for small number of users and files.

\section{Discussions}
To show the usefulness and applicability of the proposed model, i.e., the 3-user D2D coded caching with 2RR1S, and its optimal scheme given in Section \ref{Ach}, which we call the \emph{2RR1S scheme}, we adapt the 2RR1S scheme to three coded caching scenarios. 

\subsection{Scenario 1: the traditional 3-user D2D Coded Caching Problem}\label{0420d}

In the traditional 3-user D2D coded caching problem \cite{7342961}, it is assumed that in the delivery phase, all three users make file requests and they all participate in the signal transmission as senders. The following lemma illustrates the relationship between the 3-user D2D coded caching with 2RR1S and the traditional 3-user D2D coded caching problem.
\begin{Lem} \label{0420b}
Any achievable scheme for the 3-user D2D coded caching with 2RR1S can be exploited as a baseline scheme to come up with an achievable scheme for the traditional 3-user D2D coded caching problem.
\end{Lem}
%
%
\begin{proof}
Take a caching and delivery scheme for the 3-user D2D coded caching with 2RR1S that satisfies (\ref{eq:1}), denoted as Scheme $A$, and suppose it requires the split of each file into $L$ subfiles, i.e., $W_n=(W_{n,1}, W_{n,2}, \cdots, W_{n,L})$, $n \in [N]$. Then, further split each subfile into two  parts of equal sizes, denoted as Part $(a)$ and Part $(b)$, respectively, i.e., $W_{n,l}=(W_{n,l}^{(a)}, W_{n,l}^{(b)})$, $l \in [L], n \in [N]$. The caching scheme for the traditional 3-user D2D coded caching problem is the same as that of Scheme $A$. In terms of the delivery scheme, when User $k$ requests file $W_{d_k}$, $k=1,2,3$,  User 1 uses the delivery scheme of scheme $A$ for the request vector $(0, d_2, d_3)$, acting on Part (a) of each file, User 2 uses the delivery scheme of Scheme $A$ for the request vector $(d_1, 0, d_3)$, acting on Part (a) of File $W_{d_1}$ and Part (b) of the other files, and User 3 uses the delivery scheme of scheme $A$ for the request vector $(d_1, d_2, 0)$, acting on Part (b) of each file. As a result, the achievable delivery rate for the traditional 3-user D2D coded caching problem when User $k$ requests file $W_{d_k}$, $k=1,2,3$, is $\frac{1}{2} \left(R^{(0,d_2,d_3)}(M)+R^{(d_1,0,d_3)}(M)+R^{(d_1,d_2,0)}(M) \right)$, where $R^{(0,d_2,d_3)}(M), R^{(d_1,0,d_3)}(M), R^{(d_1,d_2,0)}(M)$ denotes the delivery rates of Scheme $A$.

As an example, take a caching and delivery scheme for the 3-user D2D coded caching with 2RR1S when $M=\frac{1}{2}N$ as follows: 
split all files into two subfiles of equal size, denoted as $W_n=(W_{n,1},W_{n,2})^N_{n=1}$, i.e., $L=2$. The caching scheme is $
Z_1=(W_{n,1}\oplus W_{n,2})^N_{n=1}$, 
$Z_2=(W_{n,1})^N_{n=1}$, 
$Z_3=(W_{n,2})^N_{n=1}.
$
The delivery scheme is
$
X^{(0,d_2,d_3)}=\{ W_{d_2,1} \oplus W_{d_2,2}, W_{d_3,1} \oplus W_{d_3,2}\}$, 
$X^{(d_1,0,d_3)}=\{W_{d_1,1}, W_{d_3,1}\}$, 
$X^{(d_1,d_2,0)}=\{W_{d_1,2}, W_{d_2,2}\}$.
It can be checked that each user's demand can be correctly decoded for the 3-user D2D coded caching with 2RR1S, and the delivery rate of $R^D(M)=1$ for any demand vector $D$. 

The above scheme can be used as a baseline scheme to come up with an achievable scheme for the traditional 3-user D2D coded caching problem as follows: further split each subfile into two equal parts, i.e., $W_n^{l}=(W_{n,l}^{(a)}, W_{n,l}^{(b)}), n \in [N], l\in [2]$. The caching scheme is the same as the above scheme, i.e., User 1 caches $(W_{n,1}^{(a)}\oplus W_{n,2}^{(a)}, W_{n,1}^{(b)}\oplus W_{n,2}^{(b)})^N_{n=1}$, 
User 2 caches $(W_{n,1}^{(a)}, W_{n,1}^{(b)})^N_{n=1}$, and User 3 caches
$(W_{n,2}^{(a)}, W_{n,2}^{(b)})^N_{n=1}.
$ The delivery scheme when User $k$ requests file $W_{d_k}$, $k=1,2,3$, is as follows: User 1 transmits $X^{(0,d_2,d_3)}$ of the above scheme acting on Part $(a)$ of the files only, i.e., $\{ W_{d_2,1}^{(a)} \oplus W_{d_2,2}^{(a)}, W_{d_3,1}^{(a)} \oplus W_{d_3,2}^{(a)}\}$, User 2 transmits  $X^{(d_1,0,d_3)}$ of the above scheme acting on Part $(a)$ of File $W_{d_1}$ and Part $(b)$ of the other files, i.e., $\{ W_{d_1,1}^{(a)}, W_{d_3,1}^{(b)}\}$, and User 3 transmits $X^{(d_1,d_2,0)}$ of the above scheme acting on Part $(b)$ of files, i.e., $\{ W_{d_1,2}^{(b)}, W_{d_2,2}^{(b)}\}$. It is easy to check that each user's demand can be correctly decoded for the traditional 3-user D2D coded caching problem, and the delivery rate is $\frac{3}{2}$ for any demand vector $(d_1,d_2,d_3)$. 
\end{proof}
%

Applying Lemma \ref{0420b}, using the 2RR1S scheme
as a baseline scheme, we obtain a scheme for the traditional 3-user D2D coded caching problem, called the \emph{rotated 2RR1S scheme, which achieves a delivery rate of $\frac{3}{2}R(M)$, where $R(M)$ is characterized by Theorem \ref{MainResult}}. We discuss its performance in the following.

\subsubsection{Better performance for $N=2$ due to coded cache placement}
The performance of the proposed rotated 2RR1S scheme is shown in Fig.~\ref{fig:Compare1} by the red solid line for  $N=2,3,4$, respectively. 
The achievable rates found in \cite{7342961} and \cite{8830435}, are denoted by the black dash-dot line and the blue dashed line, respectively. Recall that the achievable schemes in \cite{7342961} and \cite{8830435} both employ uncoded cache placement and one-shot delivery. 
It can be seen that the proposed rotated 2RR1S scheme does not offer better performance in the case of 3 or 4 files. So we focus on the case of 2 files, i.e., Fig. \ref{fig:Compare1}(a). 
\begin{figure}[htbp]
	\centerline{\includegraphics[width=17cm]{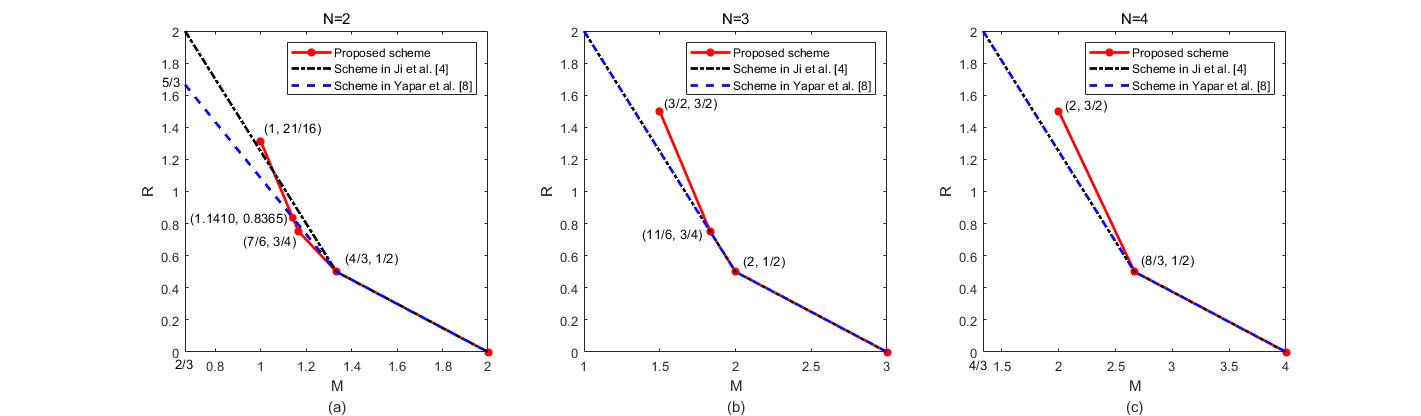}}
	\caption{Comparison of the worst-case rate achieved for three schemes for the traditional $3$-user D2D coded caching problem when  $N=2$, $N=3$, $N=4$, respectively.}
	\label{fig:Compare1}
\end{figure}

The proposed rotated 2RR1S scheme outperforms both schemes of \cite{7342961} and \cite{8830435} when $M\in \left[1.1410,\frac{4}{3} \right]$, for $N=2$, which is due to \emph{coded} cache placement. Meanwhile, the rate of the rotated 2RR1S scheme is the same as the rate of the schemes in \cite{7342961} and\cite{8830435}, when the cache capacity is large. When the cache capacity is small, the performance of the rotated 2RR1S scheme is in general loose. This is because the rotated 2RR1S scheme places a restriction on the number of senders being 1. As a result, first of all, the rotated 2RR1S scheme is only possible when $M \geq \frac{N}{2}$. Secondly, some caching schemes that are feasible for the traditional 3-user D2D coded caching problem are not included in the rotated schemes of Lemma \ref{0420b} as they require multiple senders to satisfy the decodability constraint. Hence, even though the baseline scheme is optimal for the 3-user D2D coded caching with 2RR1S, the adaptation may be sub-optimal for the traditional 3-user D2D coded caching problem. Improvements over existing D2D coded caching schemes for $N \geq 3$ are of great interest and left for future work. To do so, coded cache placement and delivery schemes that are not of the one-shot delivery nature should be considered. 


\subsubsection{Optimal for $N=2$ and medium cache size}
To further understand the performance of the proposed rotated 2RR1S scheme, we first characterize the optimal rate-memory tradeoff for the traditional 3-user D2D coded caching problem \cite{7342961} when the number of files is 2, i.e.,  $N=2$. This result is previously unknown.  
%
\begin{Theo} \label{MR_normal}
For the  traditional 3-user D2D coded caching problem for $N=2$, when $M \geq \frac{1}{3}N$, we have 
\begin{equation}
2M+R(M)\geq3, \qquad 3M+2R(M)\geq5, \qquad 3M+4R(M)\geq6\label{eq:2,3},
\end{equation}
where the corner points are $(\frac{2}{3},\frac{5}{3})$, $(1,1)$, $(\frac{4}{3},\frac{1}{2})$, $(2,0)$. Conversely, there exist caching and delivery schemes for any nonnegative $R(M)$ satisfying (\ref{eq:2,3}).
\end{Theo}
\begin{proof}
We present the achievability proof of Theorem \ref{MR_normal} in the following to show the need for coded cache placement. The converse proof is given in Appendix~\ref{AppB}.

We prove that the four corner points $(\frac{2}{3},\frac{5}{3})$, $(1,1)$, $(\frac{4}{3},\frac{1}{2})$, $(2,0)$, as stated in Theorem~\ref{MR_normal}, are achievable. The achievability proof of the three corner points $(\frac{2}{3},\frac{5}{3})$, $(\frac{4}{3},\frac{1}{2})$, and $(2,0)$ is given in \cite{8830435}, and it has been shown that uncoded cache placement is sufficient to achieve these three corner points. Thus, we only need to prove that the corner point  $(1,1)$ is achievable.

We split the two files $W_1$ and $W_2$ into 6 subfiles, which can be represented as $W_1=(A_n)^6_{n=1}$, $W_2=(B_n)^6_{n=1}$. For $M=1$, User 1 caches 
$(A_1\oplus B_1, A_2\oplus B_2, A_3, A_4, B_3, B_4)$, User 2 caches
$(A_3\oplus B_3, A_4\oplus B_4, A_5, A_6, B_5, B_6)$, and User 3 caches
$(A_5\oplus B_5, A_6\oplus B_6, A_1, A_2, B_1, B_2)$.
Note that coded cache placement is employed. 
During the delivery phase, when User $k$ requests $W_{d_k}$, then, User 1 transmits$
\lbrace W_{d_3,3}, W_{d_3,4}\rbrace$, User 2 transmits $\lbrace W_{d_1,5}, W_{d_1,6}\rbrace$, and User 3 transmits $\lbrace W_{d_2,1}, W_{d_2,2}\rbrace$.
It can be checked that each user can decode its requested file and the delivery rate is $1$. 
Memory sharing between the corner points $(\frac{2}{3},\frac{5}{3})$, $(1,1)$, $(\frac{4}{3},\frac{1}{2})$ and $(2,0)$ shows that  (\ref{eq:2,3}) in Theorem~\ref{MR_normal} is achievable.
\end{proof}
%
%
%
\begin{figure}[htbp]
	\centerline{\includegraphics[width=9cm]{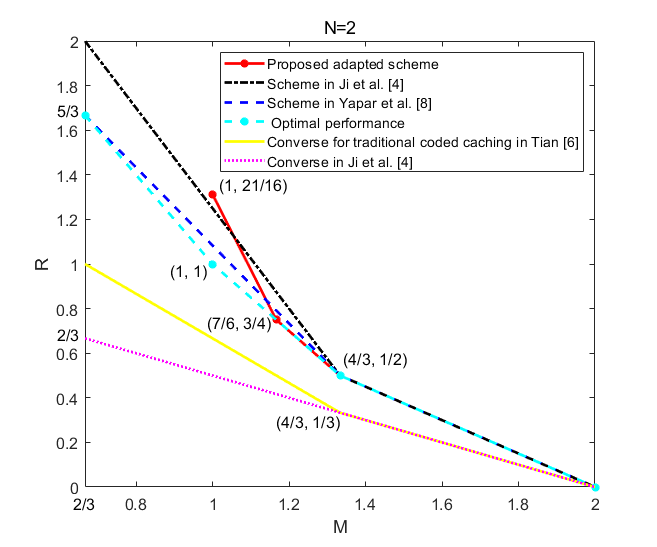}}
	\caption{Achievable performance and converse results for the traditional 3-user D2D coded caching problem when $N=2$.}
	\label{fig:Comparison}
\end{figure}
%

Now that we have found the optimal performance of the traditional 3-user D2D coded caching problem for $N=2$, i.e., Theorem \ref{MR_normal}, we may plot it in Fig.~\ref{fig:Comparison}, denoted by the cyan dotted dashed line. We further compare it to the performance of existing achievability schemes.  
The performance of our rotated 2RR1S scheme is denoted by the dotted red solid line, the scheme of \cite{7342961} is denoted by the black dash-dot line, and the scheme of \cite{8830435} is denoted by blue dashed line. As can be seen, the proposed rotated 2RR1S scheme is in fact optimal for $M\in \left[\frac{7}{6},\frac{4}{3} \right]$ in the case of $N=2$. Hence, when the cache size is in the range of $M\in \left[\frac{7}{6},\frac{4}{3} \right]$, using the scheme adapted from the 2RR1S scheme, which is the optimal scheme for the 3-user D2D coded caching with 2RR1S, does not cause any performance loss. We further observe that neither the existing schemes of \cite{7342961} and \cite{8830435}, nor the rotated 2RR1S scheme is optimal for the cache size of $M \in \left[\frac{2}{3},\frac{7}{6} \right]$. For this cache size range, the scheme used in proving Theorem \ref{MR_normal} is optimal. 

Next, we provide two remarks about the result of Theorem  \ref{MR_normal}. 

\begin{Remark}\label{Remark 7}
Recall that \cite{8830435} has found the optimal performance of the traditional $K$-user D2D coded caching problem for any number of files under the assumption of uncoded cache placement and one-shot delivery. Comparing \cite[Corollary 2]{8830435} with Theorem \ref{MR_normal} above, for 3 users and 2 files, we see that while the corner points $(\frac{2}{3},\frac{5}{3})$, $(\frac{4}{3},\frac{1}{2})$, and $(2,0)$ are the same, the corner point $(1,1)$ exists and is optimal when we remove the constraint of uncoded cache placement and one-shot delivery. We see in the proof of Theorem \ref{MR_normal} above that the achievability scheme of the corner point $(1,1)$ employs coded cache placement, but is still a one-shot delivery scheme. 
%
%
\end{Remark}

\begin{Remark}
In terms of existing converse results, the optimal result for the traditional coded caching problem in the case of 3 users and 2 files was found in \cite{e20080603}. This also serves as a converse result for the traditional 3-user D2D coded caching problem and is plotted by the yellow solid line in Fig. \ref{fig:Comparison}. The converse results derived in \cite{7342961} is denoted by the purple dotted line. From Fig.~\ref{fig:Comparison}, we can see that both converse results are rather loose, compared to the optimal performance found in Theorem \ref{MR_normal}. 
\end{Remark}

\subsection{Scenario 2: the 3-user request-random D2D coded caching}\label{LO}
In this section, we consider a 3-user D2D coded caching scenario, where in the delivery phase, each user makes a file request randomly and independently with the same probability $p$, and all three users participate in sending the delivery signal to satisfy the requests. This model is applicable in practice because not all users have file requests at the delivery phase, or some users have file requests much later than other users, for example, after the delivery phase of the other uses have been completed. We call this problem the \emph{request-random D2D coded caching}. 
%
%
%
We denote the number of users requesting files as $r$, and since
user requests happen at the beginning of the delivery phase, the server does not know the value of $r$ during the placement phase.


The performance of interest is the \emph{average worst-case delivery rate}, which is defined as follows. For a given caching and delivery scheme designed for the 3-user request-random D2D coded caching, for each $r \in [0,3]$, denote $R_r'(M)$ as the maximum delivery rate, where the maximum is over all possible size-$r$ requester sets and all possible file demands, i.e., $[N]^r$. The average worst-case delivery rate of the scheme is defined as
\begin{align}
	\bar{R}'^p(M)=\sum_{x=0}^{3} \text{Pr} \left[r=x\right]R_x'(M), \label{0423b}
\end{align}
where $\text{Pr} \left[r=x\right]=\binom{3}{x}p^x(1-p)^{3-x}$.
%
%
%

We first adapt the 2RR1S scheme to the request-random D2D coded caching problem, and the adapted version is called the \emph{request-random 2RR1S scheme}. Its caching scheme is the same as the 2RR1S scheme. As for the delivery scheme, when $r=2$, we use the delivery scheme of the 2RR1S scheme. 
When $r=3$, we use the delivery scheme of the rotated 2RR1S scheme described in Section \ref{0420d}. When $r=1$, 
we use the delivery scheme of the 2RR1S scheme, assuming a fake requester who requests the same file as the true requester. Furthermore, if some parts of the delivery signal are solely for the benefit of the fake requester, these are not transmitted. To give an example, when the number of files is $N=2$, the cache size is $M=1$, according to the 2RR1S scheme, each file is split into 8 subfiles, i.e., $W_1=(A_n)^8_{n=1}$ and $W_2=(B_n)^8_{n=1}$, and the caching scheme is the same as that of the 2RR1S scheme, given in Table \ref{table 1}. In the delivery phase, when the request vector is $D=(0,0,1)$, i.e., $r=1$, and Users 1 and 2 do not request, while User 3 requests file $W_1$, we assume that User 2 is a fake requester, who also requests file $W_1$, and User 1 is the designated sender. Thus,  for the new fake request vector $D'=(0,1,1)$, the 2RR1S scheme dictates that the signal transmitted by User 1 is $X^{(0,1,1)'}=\{A_7\oplus A_8, A_4, B_4, A_5, B_5, A_1\oplus B_2, A_2\oplus B_1\}$, as indicated in (\ref{0421a}). However, among these,  the subfiles $B_4, A_2\oplus B_1$ only serve the fake requester User 2. Therefore,  they are not transmitted and the true signal transmitted for the request-random 2RR1S scheme is $\{A_7\oplus A_8, A_4, A_5, B_5, A_1\oplus B_2\}$, and the achievable rate is $\frac{5}{8}$.

Accordingly, the performance of the request-random 2RR1S scheme is: 
1) when $r=1$: for $N=2$, the corner points are $(M,R'_1(M))=(1,\frac{5}{8})$, $(\frac{7}{6},\frac{1}{2})$, $(\frac{4}{3},\frac{1}{3})$, $(2,0)$; for $N=3$, the corner points are $(M,R'_1(M))=(\frac{3}{2},\frac{1}{2})$, $(\frac{11}{6},\frac{1}{2})$, $(2,\frac{1}{3})$, $(3,0)$; and for $N\geq 4$, the corner points are $(M,R'_1(M))=(\frac{1}{2}N,\frac{1}{2})$, $(\frac{2}{3}N,\frac{1}{3})$, $(N,0)$. The achievable rate $R'_1(M)$ is the lower convex envelope of these corner points. 
2) when $r=2$: $R'_2(M)$ is the same as $R(M)$ characterized in Theorem \ref{MainResult}.
3) when $r=3$: $R'_3(M)$ is the same as $\frac{3}{2}R(M)$, where $R(M)$ is characterized  in Theorem \ref{MainResult}.Then, the average worst-case delivery rate of the request-random 2RR1S scheme can be obtained via (\ref{0423b}). 

Existing schemes, developed for the traditional 3-user D2D coded caching problem, such as that of \cite{8830435}, can also be adapted to the request-random D2D coded caching scenario. More specifically, let Scheme $B$ be a scheme for the traditional 3-user D2D coded caching problem. Then, we adapt scheme $B$ for the 3-user request-random D2D coded caching as follows: the caching scheme is the same as that of scheme $B$. When $r=3$, the delivery scheme is the same as that in Scheme $B$. When $r=2$, we use the delivery scheme of Scheme $B$ while assuming that the user who does not request to be a fake requester, and it requests the same file as one of the two true requesters. When $r=1$, we use the delivery scheme of Scheme $B$ while assuming that the two users who do not request are fake requesters, and they both request the same file as the one true requester. Similar to the request-random 2RR1S scheme, if some parts of the delivery signal are solely for the benefit of the fake requesters, these are not transmitted. Accordingly, the performance of the adapted scheme of \cite{8830435} is: 
1) when $r=1$: for $N\geq 2$, the corner points are $(M,R'_1(M))=(\frac{1}{3}N,\frac{2}{3})$, $(\frac{2}{3}N,\frac{1}{3})$, $(N,0)$. The achievable rate $R'_1(M)$ is the lower convex envelope of the corner points. 
2) when $r=2$: for $N\geq 2$, the corner points are $(M,R'_2(M))=(\frac{1}{3}N,\frac{4}{3})$, $(\frac{2}{3}N,\frac{1}{2})$, $(N,0)$. The achievable rate $R'_2(M)$ is the lower convex envelope of the corner points. 
3) when $r=3$: $R'_3(M)$ is the same as rate characterized in \cite[Corollary 2]{8830435}.
Then, the average worst-case delivery rate of the adapted scheme of \cite{8830435} can be obtained via (\ref{0423b}).

\begin{figure}[htbp]
	\centerline{\includegraphics[width=9cm]{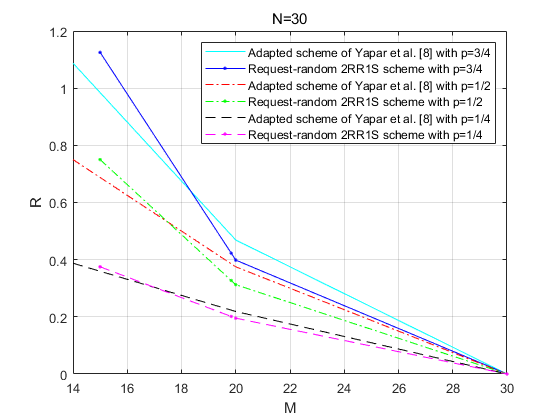}}
	\caption{Comparison of the average worst-case rate achieved for adapted schemes for the 3-user request-random D2D coded caching problem when $N = 30$.}
	\label{fig:Comparison_r}
\end{figure}

In Fig. \ref{fig:Comparison_r}, for $N=30$, we compare the average worst-case delivery rate of the request-random 2RR1S scheme and the adapted scheme of \cite{8830435}, which is the best-known achievable scheme for the traditional 3-user D2D coded caching problem.
%
%
%
The performance of our proposed request-random 2RR1S scheme is given by the dotted blue solid line when $p=\frac{3}{4}$, the dotted green dash-dot line when $p=\frac{1}{2}$, and the dotted purple dashed line when $p=\frac{1}{4}$. 
The performance of the adapted scheme of \cite{8830435},  is given by the cyan solid line when $p=\frac{3}{4}$, the red dash-dot line when $p=\frac{1}{2}$, and the black dashed line when $p=\frac{1}{4}$.
It can be seen that our proposed request-random 2RR1S scheme outperforms the adapted scheme of \cite{8830435} when the cache size is medium to large. 
Through numerical evaluation, we find that the gain can be as high as 17\% when $p=0.59$ and the cache size is $M=\frac{2}{3}N$. 

\subsection{Scenario 3: $K$-user D2D coded caching with $K-s$ random requesters and $s$ senders}\label{L1}
\vspace{-0.1in}
In this subsection, we study the $K$-user D2D coded caching scenario, where in the delivery phase, any $K-s$ out of the $K$ users make a file request, and the $s$ users who do not make file requests are the designed senders. We call this problem the \emph{$K$-user D2D coded caching with $K-s$ random requesters and $s$ senders}. The 3-user D2D coded caching with 2RR1S problem studied in Section \ref{SecNan01} is a special case of this model where $K=3$ and $s=1$. Inspired by the 2RR1S scheme, we have the following achievability result. 
\begin{Theo} \label{Theo3}
	For the $K$-user D2D coded caching with $K-s$ random requesters and $s$ senders, the following memory and worst-case rate tradeoff points $(M,R(M))$, are achievable for $N \geq 2$
	\begin{equation}
		\left(\frac{1}{s+1}N,\frac{s}{s+1}\min\{N,K-s\} \right),\qquad\left(\frac{K-1}{K}N,\frac{1}{K} \right),\qquad \left(N,0\right).\qquad\qquad\nonumber
	\end{equation}
\end{Theo}
For other memory sizes, the memory and worst-case rate tradeoff point can be obtained by the memory sharing of the corner points above.
\begin{proof}
	Similar to the 2RR1S scheme in Section~\ref{Ach}, the scheme of corner point $(\frac{1}{s+1}N,\\ \frac{s}{s+1}\min\{N,K-s\})$ employs MDS code and the scheme of corner point $(\frac{K-1}{K}N,\frac{1}{K})$ employs the MAN uncoded symmetric placement. More details of the proof is given in Appendix~\ref{0815}.   
\end{proof}
\begin{Remark}
	For the special case where $K=3$ and $s=1$, Theorem \ref{Theo3} recovers the achievability result of Theorem \ref{MainResult} Point (1). 		
\end{Remark}

\section{Conclusions}\label{Conclusions}
In this paper, we studied a new model called the 3-user D2D coded caching with 2RR1S, where during the delivery phase, any two out of the three users will make file requests and the user who does not make any file request will be the designated sender. We first characterized the optimal rate-memory tradeoff of the 3-user D2D coded caching with 2RR1S for any number of files  by proving matching converse and achievability results. It is shown that coded cache placement is needed to achieve the optimal performance. Next, using the optimal achievable scheme for the 3-user D2D coded caching with 2RR1S as a base scheme, we proposed a new achievable scheme for the 3-user D2D coded caching problem which involves coded cache placement. The new achievable scheme outperforms existing schemes when the number of files is 2 and the cache size is medium. We further characterized the optimal rate-memory tradeoff of the 3-user D2D coded caching when the number of files is 2. In doing so, we showed that the proposed new achievable scheme is in fact optimal when the cache size is medium. Comparing to existing works which focus on schemes of uncoded cache placement, we characterized the amount of performance gain enabled by allowing coded cache placement. 
Further, we adapted the optimal achievable scheme for the 3-user D2D coded caching with 2RR1S to the 3-user request-random D2D coded caching problem and found that the adapted scheme outperforms the existing schemes  when the cache size is medium to large.
Lastly, we extended the 3-user D2D coded caching with 2RR1S to the $K$-user D2D coded caching with $K-s$ random requesters and $s$ senders problem, and obtained an achievability result for this problem, which is inspired by the 2RR1S scheme and employs MDS code.

\appendices
\section{Proof of Converse for Theorem \ref{MainResult}: $N=2,3$} \label{0420a}
\subsubsection{$N=2$}\label{ProofN=2}
The bound $M+2 R(M) \geq 2$ has been proved in (\ref{22070401}), which is valid for any $N \geq 2$. We now proceed to first prove the bound $18M+8R(M)\geq 25$. First, we notice the following proposition for the problem under consideration.

\begin{Proposition}\label{Proposition 1}
For our system model with $N\geq 2$, the entropy of users' caches must satisfy
\begin{equation}\label{eq:28}
H(Z_f,A)+H(Z_g,A)\geq H(A)+H(W_{[N]}),\qquad  \forall f\ne g, \quad f,g\in[K],
\end{equation}
where $A$ can be any set of random variables of files $W_1,\cdots, W_N$, cache contents $Z_1, Z_2, Z_3$ and transmitted signals $X^{(0,d_2,d_3)}$, $X^{(d_1,0,d_3)}$, $X^{(d_1,d_2,0)}$, $d_1, d_2, d_3 \in [N]$ in our system model.
\end{Proposition}
\begin{proof}
For $\forall f,g\in[1:3]$, $f\ne g$, due to the sub-modular property of the entropy function, we have
\begin{align}
H(Z_f,A)+H(Z_g,A)\geq& H(A)+H(Z_f,Z_g,A)
\nonumber\\
= &H(A)+H(W_{[N]},A)
\label{eq:29.5}\\
= &H(A)+H(W_{[N]}),\label{eq:30}
\end{align}
where (\ref{eq:29.5}) follows from (\ref{eq:2}) and (\ref{eq:30}) follows from (\ref{1711}), (\ref{eq:0}) and the fact that $A$ can only contain files, cache contents and transmitted signals, which are all deterministic functions of the files $W_{[N]}$. 

For example,
\begin{align}
H(Z_1,W_3,X^{(2,4,0)})+H(Z_2,W_3,X^{(2,4,0)})\overset{(a)}{\geq}& H(W_3,X^{(2,4,0)})+H(Z_1,Z_2,W_3,X^{(2,4,0)})
\nonumber\\
\overset{(b)}{=} &H(W_3,X^{(2,4,0)})+H(W_{[N]},W_3,X^{(2,4,0)})
\nonumber\\
= &H(W_3,X^{(2,4,0)})+H(W_{[N]},X^{(2,4,0)})
\nonumber\\
\overset{(c)}{=} &H(W_3,X^{(2,4,0)})+H(W_{[N]},Z_3,X^{(2,4,0)})
\nonumber\\
\overset{(d)}{=} &H(W_3,X^{(2,4,0)})+H(W_{[N]},Z_3)
\nonumber\\
\overset{(e)}{=} &H(W_3,X^{(2,4,0)})+H(W_{[N]}),
\nonumber
\end{align}
where (a) follows from the sub-modular property of the entropy function, (b) follows from (\ref{eq:2}), (c) and (e) follows from (\ref{1711}), and (d) follows from (\ref{eq:0}).
\end{proof}

To prove the bound $18M+8R(M)\geq 25$, we use the following lemma, whose proof is given in Appendix~\ref{AppC}. We note here that the proof benefited from the computer-aided discovery of outer bounds \cite{e20080603} and Proposition \ref{Proposition 1}.
\begin{Lem} \label{LemmaNan02}
For $N=2$,  we have the following results:
\begin{align}
8M+6R(M)+2H(X^{(1,0,1)})-4H(W_1, X^{(0,1,2)}, X^{(1,0,1)}) & \geq 8, \label{015}\\
4M+2R(M)-2H(X^{(0,1,1)})+2H(W_1, X^{(0,1,1)})-H(W_2, X^{(0,1,2)}, X^{(1,0,2)}) &\geq 5,  \label{016}\\
6M-3H(X^{(1,0,2)}, W_1, X^{(0,1,2)})+3H(X^{(0,1,2)}, X^{(1,0,2)})  & \geq 6. \label{017}
\end{align}
\end{Lem}
Adding (\ref{015})-(\ref{017}) together, we have
\begin{align}
& \hspace{-0.2in}18M+8R(M)-3H(X^{(1,0,2)}, W_1, X^{(0,1,2)})+3H(X^{(0,1,2)}, X^{(1,0,2)})  \nonumber\\
\geq& 19+4H(W_1, X^{(0,1,2)}, X^{(1,0,1)})+H(W_2, X^{0,1,2}, X^{1,0,2})-2H(W_1, X^{0,1,1}) \nonumber\\
= &19+2H(W_1, X^{(0,2,1)}, X^{(1,1,0)})+2H(W_1, X^{(2,0,1)}, X^{(1,1,0)})-2H(W_{1},X^{(1,1,0)})
\label{013}\\
& +H(W_2,X^{(0,1,2)},X^{(1,0,2)})\nonumber\\
\geq &19+2H(W_1, X^{(0,2,1)}, X^{(1,1,0)},X^{(2,0,1)})
+H(W_2,X^{(0,1,2)},X^{(1,0,2)})\nonumber\\
= &19+2(W_2,X^{(0,1,2)}, X^{(2,2,0)}, X^{(1,0,2)})+H(W_2,X^{(0,1,2)},X^{(1,0,2)})\label{020}\\
\geq & 19+3H(W_2,X^{(0,1,2)},X^{(1,0,2)}), \nonumber
\end{align}
where (\ref{013}) and (\ref{020}) both follow from the property of file-index-symmetric schemes. Hence, we have
$
18M+8R(M) \geq 19+3H(W_2,X^{(0,1,2)},X^{(1,0,2)})+3H(X^{(1,0,2)}, W_1, X^{(0,1,2)})-3H(X^{(0,1,2)},\break X^{(1,0,2)}) 
\geq  19+3 H(W_1, W_2) =25, 
$
which completes the proof of the bound $18M+8R(M)\geq 25$.

Next, we proceed to prove the bound $3M+3R(M)\geq 5$. We have the following lemma, whose proof is given in Appendix~\ref{AppD}. We note here that the proof benefited from the computer-aided discovery of outer bounds \cite{e20080603}.
\begin{Lem} \label{LemmaNan03}
For $N=2$, we have the following results:
\begin{align}
M+R(M)+H(W_1,X^{(2,0,1)}) & \geq 3, \label{021}\\
2M+2R(M)-H(W_1,X^{(1,0,2)}) &\geq 2. \label{022}
\end{align}
\end{Lem}
Then, adding (\ref{021}) and (\ref{022}) together, we have
\begin{align}
3M+3R(M)+H(W_1,X^{(2,0,1)})-H(W_1,X^{(1,0,2)})&= 3M+3R(M) \geq 5, \label{031}
\end{align}
where the equality in (\ref{031}) follows from the property of user-index-symmetric schemes.
Thus, (\ref{eq:6}) in Theorem~\ref{MainResult} is proved.

\subsubsection{$N=3$}\label{ProofN=3}
The bound $M+3 R(M) \geq 3$ has been proved in (\ref{22070401}), which is valid for any $N \geq 2$. We now proceed to first prove the bound $6M+4R(M)\geq 13$ by using the following lemma, whose proof is in Appendix~\ref{AppE}. We note here that the proof benefited from the computer-aided discovery of outer bounds \cite{e20080603} and Proposition \ref{Proposition 1}.
\begin{Lem} \label{LemmaNan04}
For $N=3$, we have the following results:
\begin{align}
2M+2R(M)-H(W_1,X^{(0,1,3)},X^{(1,0,2)}) & \geq 3,\label{032}\\
4M+2R(M)+H(W_1,X^{(0,1,3)},X^{(1,0,2)}) & \geq 10. \label{033}
\end{align}
\end{Lem}
Adding (\ref{032}) and (\ref{033}) together, we have
$
6M+4R(M)-H(W_1,X^{(0,1,3)},X^{(1,0,2)})+H(W_1,\break X^{(0,1,3)},X^{(1,0,2)})=6M+4R(M)\geq13. 
$

Next, we proceed to prove the bound $3M+3R(M)\geq 7$. We have the following lemma, whose proof is given in Appendix~\ref{AppF}. We note here that the proof benefited from the computer-aided discovery of outer bounds \cite{e20080603} and Proposition \ref{Proposition 1}.
\begin{Lem} \label{LemmaNan05}
For $N=3$, we have the following results:
\begin{align}
M+R(M)+H(W_1,X^{(0,1,3)},X^{(1,0,2)}) & \geq 4,\label{070}\\
2M+2R(M)-H(W_1,X^{(0,1,3)},X^{(1,0,2)}) &\geq 3. \label{071}
\end{align}
\end{Lem}
Adding (\ref{070}) and (\ref{071}) together, we have
$
3M+3R(M)+H(W_1,X^{(0,1,3)},X^{(1,0,2)})-H(W_1,X^{(0,1,3)},\break X^{(1,0,2)}) =3M+3R(M)\geq7. \nonumber
$
Thus, (\ref{eq:7}) in Theorem~\ref{MainResult} is proved.

\section{Converse Proof of Theorem \ref{MR_normal}} \label{AppB}
For the 3-user D2D coded caching problem, we denote the request vector as $D\stackrel{\triangle}{=}(d_1,d_2,d_3)$, $d_1,d_2,d_3 \in [N]$. Each user sends a signal that is received without error by the other two users. The transmitted signal generated from User $k$ is denoted as $X_k^{D}$, consisting of $R_k^D(M)$ bits. Hence, we have:
	\begin{align}
		H(X_k^{D}|Z_k)=0. \label{eq:17}
	\end{align}
	Similar to the 3-user coded caching with 2RR1S problem, (\ref{1711}) and decodability constraint are applicable here. The decodability constraint means that it is required that the caching and delivery scheme must satisfy the constraint of correct decoding, i.e., 
	\begin{equation}
		\begin{aligned}\label{eq:18}
			H\left(W_{d_1}|Z_{1},X_2^{(d_1,d_2,d_3)},X_3^{(d_1,d_2,d_3)} \right)=0,\\
			H\left(W_{d_2}|Z_{2},X_1^{(d_1,d_2,d_3)},X_3^{(d_1,d_2,d_3)} \right)=0,\\
			H\left(W_{d_3}|Z_{3},X_1^{(d_1,d_2,d_3)},X_2^{(d_1,d_2,d_3)} \right)=0.
		\end{aligned}
	\end{equation}
Combining (\ref{eq:17}) and (\ref{eq:18}), one can decode any file by knowing the cache of all users, which implies that we are interested in the case where $3M\geq N$ and
	\begin{equation}
		H(W_{[N]}|Z_1,Z_2,Z_3)=0. \nonumber
	\end{equation}
$R(M)$ is defined as the worst-case delivery rate of all users, i.e., $	R(M)=\max_D  R^D_1(M)+R^D_2(M)+R^D_3(M)$.
	
	Moreover, we may define the file-index-symmetric and user-index-symmetric schemes for the 3-user D2D coded caching problem and again have the result of Lemma \ref{Lemma0} for this model also. Hence, for 3-user D2D coded caching problems, it is sufficient to consider caching and delivery schemes that satisfy both user-index symmetry and file-index symmetry. 

We first derive the following proposition for the 3-user D2D coded caching problem.

\begin{Proposition}\label{Proposition 2}
	For the 3-user D2D coded caching problem, when $N\geq 2$, the following equation must be satisfied for any caching and delivery scheme: 
	\begin{align}
		3M+2NR(M)\geq 3N. \nonumber
	\end{align}
\end{Proposition}
\begin{proof}
	We write the chain of inequalities as
	\begin{align} \label{eq:25}
		3M+2NR(M)
		\geq&H(Z_{1})+H(Z_{2})+H(Z_{3})+2N[H(X_1^{(1,2,2)})+H(X_2^{(1,2,2)})+H(X_3^{(1,2,2)})]. 
	\end{align}
	Observe that
	\begin{align}
		&H(Z_{1})+NH(X_2^{(1,2,2)})+NH(X_3^{(1,2,2)})
		\nonumber\\
		= &H(Z_{1})+H(X_2^{(1,2,2)})+H(X_3^{(1,2,2)})+\sum_{i=2}^N[H(X_2^{(i,1,1)})+H(X_3^{(i,1,1)})]
		\label{Nan40}\\
		\geq& H(Z_{1})+H(X_2^{(1,2,2)})+H(X_3^{(1,2,2)})+H(X_2^{([2:N],1,1)})+H(X_3^{([2:N],1,1)})
		\nonumber\\
		\geq& H(Z_{1}, X_2^{(1,2,2)}, X_3^{(1,2,2)}, X_2^{([2:N],1,1)}, X_3^{([2:N],1,1)})
		\nonumber\\
		= & H(Z_{1}, X_2^{(1,2,2)}, X_3^{(1,2,2)}, X_2^{([2:N],1,1)}, X_3^{([2:N],1,1)},W_{[N]})
		\nonumber\\
		= & N, \label{eq:26}
	\end{align}
	where $X_2^{([2:N],1,1)}=\{X_2^{(2,1,1)},...,X_2^{(N,1,1)}\}$, $X_3^{([2:N],1,1)}=\{X_3^{(2,1,1)},...,X_3^{(N,1,1)}\}$, (\ref{Nan40}) follows from the property of file-index-symmetric schemes, and (\ref{eq:26}) follows from (\ref{1711}), (\ref{eq:17}) and (\ref{eq:18}).
	
	Similarly, we have 
	\begin{align}
		H(Z_{2})+NH(X_1^{(1,2,2)})+NH(X_3^{(1,2,2)}) &\geq N, \label{Nan100}\\
		H(Z_{3})+NH(X_1^{(1,2,2)})+NH(X_2^{(1,2,2)}) & \geq N. \label{Nan101}
	\end{align}
	From (\ref{eq:25}), (\ref{eq:26}), (\ref{Nan100}) and (\ref{Nan101}), we have proved Proposition~\ref{Proposition 2}.
\end{proof}
The bound $3M+4R(M)\geq 6$ is a special case of Proposition~\ref{Proposition 2} when $N=2$. Hence, we are left to prove the bounds $2M+R(M)\geq 3$ and $3M+2R(M)\geq 5$.

We first prove the bound $2M+R(M)\geq 3$. Due to the property of file-index-symmetric and user-index-symmetric schemes for the 3-user D2D coded caching model, we have
	\begin{equation}
		\begin{aligned}\label{eq:36}
			H(X_1^{(1,1,2)})=H(X_1^{(1,2,1)})&=H(X_2^{(1,1,2)})=H(X_2^{(1,2,2)})=H(X_3^{(1,2,1)})=H(X_3^{(1,2,2)}),\\
			&H(X_1^{(1,2,2)})=H(X_2^{(1,2,1)})=H(X_3^{(1,1,2)}).
		\end{aligned}
	\end{equation}

Then, we have the following lemma, whose proof is given in Appendix~\ref{AppG}. We note here that the proof benefited from the computer-aided discovery of outer bounds \cite{e20080603}.

\begin{Lem} \label{LemmaNan06}
	In the 3-user D2D coded caching problem, when $N=2$, we have the following results:
	\begin{align}
		M+R(M) \geq & H(X_1^{(1,2,2)})-H(W_1,X_1^{(1,2,1)},X_2^{(1,2,1)},X_3^{(1,2,1)})-2H(W_1,Z_1,Z_2)\label{0713}\\
		& +3H(W_1,W_2)+H(W_1), \nonumber\\
		2M+R(M) \geq & H(X_1^{(1,2,1)},X_2^{(1,2,1)})+H(Z_1,X_2^{(1,2,1)},X_3^{(1,2,1)}) \label{0714}\\
		& -H(W_1,X_1^{(1,2,1)},X_2^{(1,2,1)},X_3^{(1,2,1)})+H(W_1,W_2), \nonumber\\
		3M+R(M) \geq & 2H(Z_1)+H(Z_1,X_2^{(1,1,2)},X_3^{(1,1,2)})+H(X_1^{(1,1,2)}). \label{0715}
	\end{align}
\end{Lem}

Using (\ref{eq:36}) to combine (\ref{0713}), (\ref{0714}) and (\ref{0715}), we have
\begin{align}
	&6M+3R(M)
	\geq 2[H(Z_1,X_2^{(1,1,2)},X_3^{(1,1,2)})+H(Z_2,X_1^{(1,1,2)},X_3^{(1,1,2)})\nonumber\\
	& -H(W_1,X_1^{(1,2,1)},X_2^{(1,2,1)},X_3^{(1,2,1)})]
	-2H(W_1,Z_1,Z_2)+4H(W_1,W_2)+H(W_1)
	\label{Nan45.99}\\
	= &2[H(Z_1,W_1,X_1^{(1,2,1)},X_2^{(1,2,1)},X_3^{(1,2,1)})+H(Z_2,W_1,X_1^{(1,2,1)},X_2^{(1,2,1)},X_3^{(1,2,1)})
	\label{Nan45.995}\\
	& -H(W_1,X_1^{(1,2,1)},X_2^{(1,2,1)},X_3^{(1,2,1)})]-2H(W_1,Z_1,Z_2)+4H(W_1,W_2)+H(W_1)
	\nonumber\\
	\geq &2H(Z_1,Z_2,W_1,X_1^{(1,2,1)},X_2^{(1,2,1)},X_3^{(1,2,1)})-2H(W_1,Z_1,Z_2)+4H(W_1,W_2)+H(W_1)
	\label{Nan46}\\
	\geq &4H(W_1,W_2)+H(W_1)= 9,
	\label{eq:40}
\end{align}
where (\ref{Nan45.99}) follows from the property of user-index-symmetric and file-index-symmetric schemes,
(\ref{Nan45.995}) follows from (\ref{eq:17}) and (\ref{eq:18}).
(\ref{Nan46}) follows from the sub-modular property of the entropy function.
The chain of inequalities (\ref{eq:40}) proves the bound $2M+R(M)\geq 3$.

Next, we proceed to prove the bound $3M+2R(M)\geq 5$. We have the following lemma, whose proof is given in Appendix~\ref{AppH}. We note here that the proof benefited from the computer-aided discovery of outer bounds \cite{e20080603}.
\begin{Lem} \label{LemmaNan07}
	In the 3-user D2D coded caching problem, when $N=2$, we have the following results:
	\begin{align}
		3M+2R(M)-H(W_1,Z_1)-H(W_1,X_1^{(1,1,2)},X_2^{(1,1,2)}) \geq & 2, \label{0716}\\
		H(W_1,Z_1)+H(W_1,X_1^{(1,1,2)},X_2^{(1,1,2)}) \geq & 3. \label{0717}
	\end{align}
\end{Lem}
Adding (\ref{0716}) and (\ref{0717}) together, we have
$
3M+2R(M)\geq 5.
$

%
%
Thus, (\ref{eq:2,3}) in Theorem~\ref{MR_normal} is proved.

\section{Proof of Lemma \ref{LemmaNan02}} \label{AppC}
We will first prove (\ref{015}).
\begin{align}
	8M+6R(M)+2H(X^{(1,0,1)}) \geq& 4H(Z_1)+4H(Z_2)+4H(X^{(0,1,2)})\nonumber\\
	&+2H(X^{(1,0,1)}) +2H(X^{(1,0,1)})\nonumber\\
	\geq& 4H(Z_1, X^{(1,0,1)})+4H(Z_2, X^{(0,1,2)})\nonumber\\
	=&4H(Z_1, X^{(1,0,1)}, X^{(0,1,2)})+4H(Z_2, X^{(0,1,2)}, X^{(1,0,1)})\label{006}\\
	=&4H(Z_1, X^{(1,0,1)}, X^{(0,1,2)}, W_1)\label{007}\\
	&+4H(Z_2, X^{(0,1,2)}, X^{(1,0,1)}, W_1) \nonumber\\
	\geq& 4H(W_1, X^{(0,1,2)}, X^{(1,0,1)})+4H (W_1,W_2), \label{008}
\end{align}
where (\ref{006}) follows from the fact that $X^{(0,1,2)}$ is a deterministic function of $Z_1$ and $X^{(1,0,1)}$ is a deterministic function of $Z_2$, (\ref{007}) follows from the fact that knowing $(Z_1, X^{(1,0,1)})$ can decode $W_1$ and knowing $(Z_2, X^{(0,1,2)})$ can decode $W_1$, and (\ref{008}) follows from (\ref{eq:28}).
Thus, (\ref{015}) is proved.

Next, we will prove (\ref{016}). We have
\begin{align}
	4M+2R(M)-2H(X^{(0,1,1)}) \geq& 2H(Z_1)+2H(Z_2)+2H(X^{(0,1,2)})-2H(X^{(0,1,1)}) \nonumber\\
	=& 2H(Z_1, X^{(0,1,1)})+2H(Z_2)+2H(X^{(0,1,2)})\label{000}\\
	&-2H(X^{(0,1,1)})\nonumber\\
	\geq& 2H(Z_1, X^{(0,1,1)})+2H(Z_2,X^{(0,1,2)})-2H(X^{(0,1,1)})\nonumber\\
	\geq&  2H(W_1,Z_1, X^{(0,1,1)})-2H(W_1, X^{(0,1,1)})\label{001}\\
	&+2H(Z_2,X^{(0,1,2)}), \nonumber
\end{align}
where (\ref{000}) follows from the fact that $X^{(0,1,1)}$ is a deterministic function of $Z_1$, (\ref{001}) follows the sub-modular property of the entropy function. Hence, we further have
\begin{align}
	&4M+2R(M)+2H(W_1, X^{(0,1,1)}) -2H(X^{(0,1,1)}) \nonumber\\
	\geq & 2H(W_1,Z_1, X^{(0,1,1)})+2H(Z_2,X^{(0,1,2)}) \nonumber\\
	\geq & 2H(W_1,Z_1)+2H(Z_2,X^{(0,1,2)}) \nonumber\\
	= & 2H(W_1,Z_1, X^{(0,1,2)})+2H(Z_2,X^{(0,1,2)}, W_1)\label{002}\\
	\geq & 2H(W_1,X^{(0,1,2)})+2H(W_1,W_2) \label{003}\\
	= & H(W_2, X^{(0,1,2)})+H(W_2,X^{(1,0,2)})+2H(W_1,W_2) \label{004}\\
	\geq & H(W_2, X^{(0,1,2)}, X^{(1,0,2)})+H(W_2)+2H(W_1,W_2), \label{005}
\end{align}
where (\ref{002}) follows from the fact that $X^{(0,1,2)}$ is a deterministic function of $Z_1$ and knowing $(Z_2,X^{(0,1,2)})$ can decode $W_1$, (\ref{003}) follows from (\ref{eq:28})
, (\ref{004}) follows from
\begin{align}
	&H(W_1,X^{(0,1,2)}) \overset{(a)}{=}H(W_1, X^{(2,0,1)})\overset{(b)}{=}H(W_2, X^{(1,0,2)}), \nonumber\\
	&H(W_1,X^{(0,1,2)}) \overset{(c)}{=} H(W_2,X^{(0,2,1)})\overset{(d)}{=} H(W_2,X^{(0,1,2)}),
	\nonumber
\end{align}
where $(a)$ and $(d)$ follow from the property of user-index-symmetric schemes, and $(b)$ and $(c)$ follow from the property of file-index-symmetric schemes, and (\ref{005}) follows again from the sub-modular function of the entropy function. Hence, (\ref{016}) is proved.

Finally, to prove (\ref{017}), we have
\begin{align}
	6M & \geq 6 H(Z_1) \nonumber\\
	& =6 H(Z_1)+3H(X^{(1,0,2)})-3H(X^{(0,1,2)}) \nonumber\\
	& \geq 3H(Z_1)+3H(Z_1,X^{(1,0,2)})-3H(X^{(0,1,2)}) \nonumber\\
	& = 3H(Z_1, X^{(0,1,2)})+3H(Z_1,X^{(1,0,2)})-3H(X^{(0,1,2)}) \label{108}\\
	& \geq 3H(Z_1, X^{(0,1,2)}, X^{(1,0,2)})+3H(Z_1,X^{(1,0,2)})-3H(X^{(0,1,2)}, X^{(1,0,2)})  \label{009}\\
	& \geq 6H(Z_1, X^{(1,0,2)}) -3H(X^{(0,1,2)}, X^{(1,0,2)})\nonumber\\
	&=3H(Z_1, X^{(1,0,2)}) +3H(Z_2, X^{(0,1,2)})-3H(X^{(0,1,2)}, X^{(1,0,2)})\label{010}\\
	&=3H(Z_1, X^{(1,0,2)}, W_1, X^{(0,1,2)}) +3H(Z_2, X^{(0,1,2)}, W_1, X^{(1,0,2)})\label{011}\\
	&\quad-3H(X^{(0,1,2)}, X^{(1,0,2)}) \nonumber\\
	& \geq 3H(X^{(1,0,2)}, W_1, X^{(0,1,2)})+3H(W_1,W_2)-3H(X^{(0,1,2)}, X^{(1,0,2)}), \label{012}
\end{align}
where (\ref{108}) follows from the fact that $X^{(0,1,2)}$ is a deterministic function of $Z_1$, (\ref{009}) follows from the sub-modular property of the entropy function, (\ref{010}) follows from the property of user-index-symmetric schemes, (\ref{011}) follows from the fact that knowing $(Z_1, X^{(1,0,2)})$ can decode $W_1$, $X^{(0,1,2)}$ is a deterministic function of $Z_1$, knowing $(Z_2, X^{(0,1,2)})$ can decode $W_1$ and $X^{(1,0,2)}$ is a deterministic function of $Z_2$, and (\ref{012}) follows from (\ref{eq:28}). Thus, (\ref{017}) is proved.

\section{Proof of Lemma \ref{LemmaNan03}} \label{AppD}
We will first prove (\ref{021}).
\begin{align}
	M+R(M)+H(W_1,X^{(2,0,1)}) & \geq H(Z_1)+H(X^{(1,0,2)})+H(W_1,X^{(2,0,1)})\nonumber\\
	&\geq H(Z_1,X^{(1,0,2)})+H(W_1,X^{(2,0,1)})\nonumber\\
	&= H(Z_1,X^{(1,0,2)},W_1)+H(W_1,X^{(2,0,1)})\label{023}\\
	&\geq H(Z_1,W_1)+H(W_1,X^{(2,0,1)})\nonumber\\
	&\geq H(Z_1,W_1,X^{(2,0,1)})+H(W_1)\label{024}\\
	&= H(Z_1,W_1,X^{(2,0,1)},W_2)+H(W_1)\label{025}\\
	&= H(W_1,W_2)+H(W_1)=3,\label{026}
\end{align}
where (\ref{023}) follows from the fact that knowing $(Z_1, X^{(1,0,2)})$ can decode $W_1$, (\ref{024}) follows from the sub-modular property of the entropy function, (\ref{025}) follows from the fact that knowing $(Z_1, X^{(2,0,1)})$ can decode $W_2$, and (\ref{026}) follows from (\ref{1711}), (\ref{eq:0}). Hence, (\ref{021}) is proved.
\vspace{-0.26in}

Next, we will prove (\ref{022}).
\begin{align}
	&2M+2R(M)-H(W_1,X^{(1,0,2)})\nonumber\\
	\geq & H(Z_1)+H(Z_2)+H(X^{(1,0,2)})+H(X^{(0,1,2)})-H(W_1,X^{(1,0,2)})\nonumber\\
	\geq & H(Z_1,X^{(1,0,2)})+H(Z_2,X^{(0,1,2)})-H(W_1,X^{(1,0,2)})\nonumber\\
	= & H(Z_1,X^{(1,0,2)},W_1,X^{(0,2,1)})+H(Z_2,X^{(0,1,2)},W_1,X^{(1,0,2)})-H(W_1,X^{(1,0,2)})\label{027}\\
	\geq & H(Z_1,X^{(1,0,2)},W_1,X^{(0,2,1)})+H(Z_2,X^{(0,1,2)},W_1,X^{(1,0,2)},X^{(0,2,1)})\label{028}\\
	& -H(W_1,X^{(1,0,2)},X^{(0,2,1)})\nonumber\\
	\geq & H(Z_2,X^{(0,1,2)},W_1,X^{(1,0,2)},X^{(0,2,1)},W_2)\label{029}\\
	= & H(W_1,W_2)=2,\label{030}
\end{align}
where (\ref{027}) follows from the fact that knowing $(Z_1, X^{(1,0,2)})$ can decode $W_1$, $X^{(0,2,1)}$ is a deterministic function of $Z_1$, knowing $(Z_2, X^{(0,1,2)})$ can decode $W_1$ and $X^{(1,0,2)}$ is a deterministic function of $Z_2$, (\ref{028}) follows from the sub-modular property of the entropy function, (\ref{029}) follows from the fact that knowing $(Z_2, X^{(0,2,1)})$ can decode $W_2$, and (\ref{030}) follows from (\ref{1711}) and (\ref{eq:0}). Hence, (\ref{022}) is proved.

\section{Proof of Lemma \ref{LemmaNan04}} \label{AppE}
We will first prove (\ref{032}).
\begin{align}
	&2M+2R(M)-H(W_1,X^{(0,1,3)},X^{(1,0,2)}) \nonumber\\
	\geq & H(Z_1)+H(Z_2)+H(X^{(1,0,2)})+H(X^{(0,1,3)})-H(W_1,X^{(0,1,3)},X^{(1,0,2)})\nonumber\\
	\geq & H(Z_1,X^{(1,0,2)})+H(Z_2,X^{(0,1,3)})-H(W_1,X^{(0,1,3)},X^{(1,0,2)})\nonumber\\
	= & H(Z_1,X^{(1,0,2)},W_1,X^{(0,1,3)})+H(Z_2,X^{(0,1,3)},W_1,X^{(1,0,2)})-H(W_1,X^{(0,1,3)},X^{(1,0,2)})\label{034}\\
	\geq & H(W_1,W_2,W_3)+H(W_1,X^{(0,1,3)},X^{(1,0,2)})-H(W_1,X^{(0,1,3)},X^{(1,0,2)})\label{035}\\
	= & 3,\nonumber
\end{align}
where (\ref{034}) follows from the fact that knowing $(Z_1, X^{(1,0,2)})$ can decode $W_1$, $X^{(0,1,3)}$ is a deterministic function of $Z_1$, knowing $(Z_2, X^{(0,1,3)})$ can decode $W_1$ and $X^{(1,0,2)}$ is a deterministic function of $Z_2$, (\ref{035}) follows from (\ref{eq:28}). Hence, (\ref{032}) is proved.

Next, we will prove (\ref{033}). Firstly, we have
\begin{align}
	&3M+2R(M)+H(X^{(1,0,3)})\nonumber\\
	\geq & 3H(Z_1)+2H(X^{(1,0,2)})+H(X^{(1,0,3)})\nonumber\\
	\geq & 2H(Z_1,X^{(1,0,2)})+H(Z_1,X^{(1,0,3)})\nonumber\\
	= & H(Z_1,X^{(1,0,2)})+H(Z_1,X^{(1,0,2)},W_1)+H(Z_1,X^{(1,0,3)},W_1)\label{037}\\
	\geq & H(Z_1,X^{(1,0,2)})+H(Z_1,X^{(1,0,2)},W_1,X^{(1,0,3)})+H(Z_1,W_1)\label{038}\\
	= & H(Z_1,X^{(1,0,2)})+H(Z_2,X^{(0,1,2)},W_1,X^{(0,1,3)})+H(Z_1,W_1)\label{039}\\
	= & H(Z_1,X^{(1,0,2)},W_1,X^{(0,1,2)},X^{(0,1,3)})+H(Z_2,X^{(0,1,2)},W_1,X^{(0,1,3)},X^{(1,0,2)})\label{040}\\
	&+H(Z_1,W_1)\nonumber\\
	\geq & H(W_1,W_2,W_3)+H(W_1,X^{(0,1,2)},X^{(0,1,3)},X^{(1,0,2)})+H(Z_1,W_1)\label{041}\\
	= & 3+H(W_1,X^{(0,1,2)},X^{(0,1,3)},X^{(1,0,2)})+H(Z_1,W_1),\label{042}
\end{align}
where (\ref{037}) follows from the fact that knowing $(Z_1, X^{(1,0,2)})$ or $(Z_1, X^{(1,0,3)})$ both can decode $W_1$, (\ref{038}) follows from the sub-modular
property of the entropy function, (\ref{039}) follows from the property of user-index-symmetric schemes,
(\ref{040}) follows from the fact that $X^{(0,1,2)},X^{(0,1,3)}$ is a deterministic function of $Z_1$ and $X^{(1,0,2)}$ is a deterministic function of $Z_2$, (\ref{041}) follows from (\ref{eq:28}). Then, we notice that
\begin{align}
	&M-H(X^{(1,0,2)})+H(X^{(0,1,2)},X^{(0,1,3)},X^{(1,0,2)})\nonumber\\
	\geq & H(Z_1)-H(X^{(1,0,2)})+H(X^{(0,1,2)},X^{(0,1,3)},X^{(1,0,2)})\nonumber\\
	= & H(Z_1,X^{(0,1,2)},X^{(0,1,3)})-H(X^{(1,0,2)})+H(X^{(0,1,2)},X^{(0,1,3)},X^{(1,0,2)})\label{043}\\
	\geq & H(Z_1,X^{(0,1,2)},X^{(0,1,3)},X^{(1,0,2)})+H(X^{(0,1,2)},X^{(0,1,3)})-H(X^{(1,0,2)})\label{044}\\
	= & H(Z_1,X^{(1,0,2)})+H(X^{(0,1,2)},X^{(0,1,3)})-H(X^{(1,0,2)})\label{045}\\
	= & H(Z_1,X^{(1,0,3)})+H(X^{(1,0,3)},X^{(2,0,3)})-H(X^{(1,0,3)})\label{046}\\
	\geq & H(Z_1,X^{(1,0,3)},X^{(2,0,3)}),\label{047}
\end{align}
where (\ref{043}) and (\ref{045}) both follow from the fact that $X^{(0,1,2)},X^{(0,1,3)}$ is a deterministic function of $Z_1$, (\ref{044}) and (\ref{047}) both follow from the sub-modular
property of the entropy function, (\ref{046}) follows from
\begin{align}
	&H(Z_1,X^{(1,0,2)}) \overset{(a)}{=}H(Z_1, X^{(1,0,3)}), H(X^{(1,0,2)}) \overset{(b)}{=}H(X^{(1,0,3)}), \nonumber\\
	&H(X^{(0,1,2)},X^{(0,1,3)}) \overset{(c)}{=}H(X^{(2,0,1)},X^{(3,0,1)})\overset{(d)}{=}H(X^{(2,0,3)},X^{(1,0,3)}),\nonumber
\end{align}
where $(a)$, $(b)$ and $(d)$ all follow from the property of file-index-symmetric schemes, and $(c)$ follows from the property of user-index-symmetric schemes.

Using the property of user-index-symmetric schemes to combine (\ref{042}) and (\ref{047}), we have
\begin{align}
	&4M+2R(M) \nonumber\\
	\geq & 3+H(W_1,X^{(0,1,2)},X^{(0,1,3)},X^{(1,0,2)})+H(Z_1,W_1)+H(Z_1,X^{(1,0,3)},X^{(2,0,3)})\nonumber\\
	&-H(X^{(0,1,2)},X^{(0,1,3)},X^{(1,0,2)})\nonumber\\
	\geq & 3+H(W_1,W_2,X^{(0,1,2)},X^{(0,1,3)},X^{(1,0,2)})+H(Z_1,W_1)+H(Z_1,X^{(1,0,3)},X^{(2,0,3)})\label{049}\\
	& -H(W_2,X^{(0,1,2)},X^{(0,1,3)},X^{(1,0,2)})\nonumber\\
	= & 3+H(W_1,W_2,X^{(0,2,1)},X^{(0,3,1)},X^{(1,2,0)})+H(Z_1,W_1)+H(Z_1,X^{(1,0,3)},X^{(2,0,3)})\label{050}\\
	& -H(W_2,X^{(0,2,1)},X^{(0,3,1)},X^{(1,2,0)})\nonumber\\
	\geq & 3+H(W_1,W_2,X^{(0,2,1)},X^{(0,3,1)},X^{(1,2,0)},X^{(3,0,2)})+H(Z_1,W_1)\label{051}\\
	& +H(Z_1,X^{(1,0,3)},X^{(2,0,3)}) -H(W_2,X^{(0,2,1)},X^{(0,3,1)},X^{(1,2,0)},X^{(3,0,2)})\nonumber\\
	= &3+H(W_3,W_2,X^{(0,2,3)},X^{(0,1,3)},X^{(3,2,0)},X^{(1,0,2)})+H(Z_1,W_1)\label{052}\\
	& +H(Z_1,X^{(1,0,3)},X^{(2,0,3)}) -H(W_2,X^{(0,2,1)},X^{(0,3,1)},X^{(1,2,0)},X^{(3,0,2)})\nonumber\\
	\geq & 3+H(W_3,W_2,X^{(0,2,3)},X^{(0,1,3)},X^{(1,0,2)})+H(Z_1,W_1)+H(Z_1,X^{(1,0,3)},X^{(2,0,3)})\label{053}\\
	& -H(W_2,X^{(0,2,1)},X^{(0,3,1)},X^{(1,2,0)},X^{(3,0,2)}),\nonumber
\end{align}
where (\ref{049}) and (\ref{051}) both follow from the sub-modular
property of the entropy function, (\ref{050}) follows from the property of user-index-symmetric schemes, (\ref{052}) follows from the property of file-index-symmetric schemes.

Adding $H(W_1,W_2,X^{(0,2,3)},X^{(0,1,3)},X^{(1,0,2)})$ to both sides of (\ref{053}), we have
\begin{align}
	&4M+2R+H(W_1,W_2,X^{(0,2,3)},X^{(0,1,3)},X^{(1,0,2)}) \nonumber\\
	\geq & 3+H(W_3,W_2,X^{(0,2,3)},X^{(0,1,3)},X^{(1,0,2)})+H(W_1,W_2,X^{(0,2,3)},X^{(0,1,3)},X^{(1,0,2)})\nonumber\\
	& +H(Z_1,W_1) +H(Z_1,X^{(1,0,3)},X^{(2,0,3)})-H(W_2,X^{(0,2,1)},X^{(0,3,1)},X^{(1,2,0)},X^{(3,0,2)})\nonumber\\
	\geq & 3+H(W_1,W_2,W_3,X^{(0,2,3)},X^{(0,1,3)},X^{(1,0,2)})+H(W_2,X^{(0,2,3)},X^{(0,1,3)},X^{(1,0,2)})\label{054}\\
	& +H(Z_1,W_1) +H(Z_1,X^{(1,0,3)},X^{(2,0,3)})-H(W_2,X^{(0,2,1)},X^{(0,3,1)},X^{(1,2,0)},X^{(3,0,2)})\nonumber\\
	= & 3+H(W_1,W_2,W_3)+H(W_2,X^{(0,2,3)},X^{(0,1,3)},X^{(1,0,2)})+H(Z_1,W_1)\label{055}\\
	& +H(Z_1,X^{(1,0,3)},X^{(2,0,3)})-H(W_2,X^{(0,2,1)},X^{(0,3,1)},X^{(1,2,0)},X^{(3,0,2)})\nonumber\\
	= & 6+H(W_2,X^{(0,2,3)},X^{(0,1,3)},X^{(1,0,2)})+H(Z_1,W_1)+H(Z_1,X^{(1,0,3)},X^{(2,0,3)})\nonumber\\
	& -H(W_2,X^{(0,2,1)},X^{(0,3,1)},X^{(1,2,0)},X^{(3,0,2)})\nonumber\\
	= & 6+H(W_2,X^{(0,2,3)},X^{(0,1,3)},X^{(1,0,2)})+H(Z_1,W_1)+H(Z_2,X^{(0,1,3)},X^{(0,2,3)})\label{056}\\
	& -H(W_2,X^{(0,2,1)},X^{(0,3,1)},X^{(1,2,0)},X^{(3,0,2)})\nonumber\\
	= & 6+H(W_2,X^{(0,2,3)},X^{(0,1,3)},X^{(1,0,2)})+H(Z_1,W_1)\label{057}\\
	& +H(Z_2,X^{(0,1,3)},X^{(0,2,3)},X^{(1,0,2)},W_1,W_2) -H(W_2,X^{(0,2,1)},X^{(0,3,1)},X^{(1,2,0)},X^{(3,0,2)}),\nonumber
\end{align}
where (\ref{054}) follows from the sub-modular
property of the entropy function, (\ref{055}) follows from (\ref{1711}) and (\ref{eq:0}), (\ref{056}) follows from the property of user-index-symmetric schemes, (\ref{057}) follows from the fact that $X^{(1,0,2)}$ is a deterministic function of $Z_2$ and knowing $(Z_2, X^{(0,1,3)}, X^{(0,2,3)})$ can decode $W_1,W_2$.

Through further calculations, we find that
\begin{align}
	& H(W_2,X^{(0,2,3)},X^{(0,1,3)},X^{(1,0,2)})-H(W_2,X^{(0,2,1)},X^{(0,3,1)},X^{(1,2,0)},X^{(3,0,2)})\nonumber\\
	= &
	H(X^{(0,1,3)},X^{(0,2,3)},W_2,X^{(1,0,2)})-H(X^{(0,1,3)},X^{(0,2,3)},X^{(3,2,0)},W_2,X^{(1,0,2)})\label{065}\\
	\geq & H(W_2,X^{(1,0,2)})-H(X^{(3,2,0)},W_2,X^{(1,0,2)})\label{066}\\
	= & H(W_1,X^{(2,0,1)})-H(X^{(1,0,2)},W_1,X^{(0,1,3)}),\label{067}
\end{align}
where (\ref{065}) follows from the property of user-index-symmetric schemes, (\ref{066}) follows from the sub-modular
property of the entropy function, (\ref{067}) follows from
\begin{align}
	&H(W_2,X^{(1,0,2)}) \overset{(a)}{=}H(W_1,X^{(2,0,1)}), \nonumber\\
	&H(X^{(3,2,0)},W_2,X^{(1,0,2)}) \overset{(b)}{=}H(X^{(2,1,0)},W_1,X^{(3,0,1)})\overset{(c)}{=}H(X^{(1,0,2)},W_1,X^{(0,1,3)}),\nonumber
\end{align}
where $(a)$ and $(b)$ both follow from the property of file-index-symmetric schemes, and $(c)$ follows from the property of user-index-symmetric schemes.

Again, adding $H(Z_1,W_1,W_2,X^{(0,2,3)},X^{(0,1,3)},X^{(1,0,2)})$ to both sides of (\ref{057}) and applying (\ref{067}), we have
\begin{align}
	&4M+2R+H(W_1,W_2,X^{(0,2,3)},X^{(0,1,3)},X^{(1,0,2)})+H(Z_1,W_1,W_2,X^{(0,2,3)},X^{(0,1,3)},X^{(1,0,2)}) \nonumber\\
	\geq &6+H(Z_1,W_1)+H(Z_2,X^{(0,1,3)},X^{(0,2,3)},X^{(1,0,2)},W_1,W_2)\nonumber\\
	& +H(Z_1,W_1,W_2,X^{(0,2,3)},X^{(0,1,3)},X^{(1,0,2)})+H(W_1,X^{(2,0,1)})-H(X^{(1,0,2)},W_1,X^{(0,1,3)})\nonumber\\
	\geq &6+H(Z_1,W_1)+H(W_1,W_2,W_3)+H(W_1,W_2,X^{(0,2,3)},X^{(0,1,3)},X^{(1,0,2)})\label{058}\\
	& +H(W_1,X^{(2,0,1)})-H(X^{(1,0,2)},W_1,X^{(0,1,3)})\nonumber\\
	= &9+H(Z_1,W_1)+H(W_1,W_2,X^{(0,2,3)},X^{(0,1,3)},X^{(1,0,2)})+H(W_1,X^{(2,0,1)})\label{059}\\
	& -H(X^{(1,0,2)},W_1,X^{(0,1,3)}),\nonumber
\end{align}
where (\ref{058}) follows from (\ref{eq:28}).

Then, adding $H(W_1,X^{(0,1,3)},X^{(1,0,2)})$ to both sides of (\ref{059}), and removing $H(W_1,W_2,X^{(0,2,3)},X^{(0,1,3)},X^{(1,0,2)})$ and $H(Z_1,W_1,W_2,X^{(0,2,3)},X^{(0,1,3)},X^{(1,0,2)})$ from both sides of (\ref{059}), we have
\begin{align}
	&4M+2R+H(W_1,X^{(0,1,3)},X^{(1,0,2)})\nonumber\\
	\geq & 9+H(Z_1,W_1)-H(Z_1,W_1,W_2,X^{(0,2,3)},X^{(0,1,3)},X^{(1,0,2)})+H(W_1,X^{(2,0,1)})\nonumber\\
	= & 9+H(Z_1,W_1)-H(Z_1,W_2,X^{(1,0,2)})+H(W_1,X^{(2,0,1)})\label{060}\\
	= & 9+H(Z_1,W_1)-H(Z_1,W_1,X^{(2,0,1)})+H(W_1,X^{(2,0,1)})\label{061}\\
	\geq &9+H(W_1)-H(X^{(2,0,1)},W_1)+H(W_1,X^{(2,0,1)})\label{062}\\
	= & 10,\nonumber
\end{align}
where (\ref{060}) follows from the fact that $X^{(0,2,3)},X^{(0,1,3)}$ is a deterministic function of $Z_1$ and knowing $(Z_1, X^{(1,0,2)})$ can decode $W_1$, (\ref{061}) follows from the property of file-index-symmetric schemes, (\ref{062}) follows from the sub-modular
property of the entropy function. Hence, (\ref{033}) is proved.

\section{Proof of Lemma \ref{LemmaNan05}} \label{AppF}
We will first prove (\ref{070}).
\begin{align}
	&M+R(M)+H(W_1,X^{(0,1,3)},X^{(1,0,2)})\geq H(Z_1)+H(X^{(1,0,2)})+H(W_1,X^{(0,1,3)},X^{(1,0,2)})\nonumber\\
	\geq & H(Z_1,X^{(1,0,2)})+H(W_1,X^{(0,1,3)},X^{(1,0,2)}) \nonumber\\
	=& H(Z_1,X^{(1,0,2)},W_1)+H(W_1,X^{(0,1,3)},X^{(1,0,2)})\label{072}\\
	\geq &H(Z_1,W_1)+H(W_1,X^{(0,1,3)},X^{(1,0,2)})\nonumber\\
	= & H(Z_1,W_1)+H(W_1,X^{(1,0,2)})+H(W_1,X^{(0,1,3)},X^{(1,0,2)})-H(W_1,X^{(1,0,2)})\nonumber\\
	= & H(Z_1,W_1)+H(W_1,X^{(2,0,1)})+H(W_1,X^{(0,1,3)},X^{(1,0,2)})-H(W_1,X^{(1,0,2)})\label{073}\\
	\geq & H(Z_1,X^{(2,0,1)},W_1)+H(W_1)+H(Z_3,W_1,X^{(0,1,3)},X^{(1,0,2)})-H(Z_3,W_1,X^{(1,0,2)})\label{074}\\
	= &H(Z_3,X^{(1,0,2)},W_1)+H(W_1)+H(Z_3,W_1,X^{(0,1,3)},X^{(1,0,2)})-H(Z_3,W_1,X^{(1,0,2)})\label{075}\\
	= &H(W_1)+H(Z_3,W_1,X^{(0,1,3)},X^{(1,0,2)},W_2,W_3)\label{076}\\
	= &H(W_1)+H(W_1,W_2,W_3)=4,\label{077}
\end{align}
where (\ref{072}) follows from the fact that knowing $(Z_1, X^{(1,0,2)})$ can decode $W_1$, (\ref{073}) and (\ref{075}) both follow from the property of user-index-symmetric schemes, (\ref{074}) follows from the sub-modular
property of the entropy function, (\ref{076}) follows from the fact that knowing $(Z_3, X^{(0,1,3)},\break X^{(1,0,2)})$ can decode $W_2,W_3$, (\ref{077}) follows from (\ref{1711}) and (\ref{eq:0}). Hence, (\ref{070}) is proved.

Next, we will prove (\ref{071}).
\begin{align}
	&2M+2R(M)-H(W_1,X^{(0,1,3)},X^{(1,0,2)})\nonumber\\ 
	\geq & H(Z_1)+H(Z_2)+H(X^{(1,0,2)})+H(X^{(0,1,3)})-H(W_1,X^{(0,1,3)},X^{(1,0,2)})\nonumber\\
	\geq & H(Z_1,X^{(1,0,2)})+H(Z_2,X^{(0,1,3)})-H(W_1,X^{(0,1,3)},X^{(1,0,2)})\nonumber\\
	= & H(Z_1,X^{(1,0,2)},W_1,X^{(0,1,3)})+H(Z_2,X^{(0,1,3)},W_1,X^{(1,0,2)})-H(W_1,X^{(0,1,3)},X^{(1,0,2)})\label{078}\\
	\geq & H(W_1,W_2,W_3)+H(X^{(1,0,2)},W_1,X^{(0,1,3)})-H(W_1,X^{(0,1,3)},X^{(1,0,2)})\label{079}\\
	= & 3,\nonumber
\end{align}
where (\ref{078}) follows from the fact that knowing $(Z_1, X^{(1,0,2)})$ can decode $W_1$, $X^{(0,1,3)}$ is a deterministic function of $Z_1$, knowing $(Z_2, X^{(0,1,3)})$ can decode $W_1$, and $X^{(1,0,2)}$ is a deterministic function of $Z_2$, (\ref{079}) follows from (\ref{eq:28}). Hence, (\ref{071}) is proved.

\section{Proof of Lemma \ref{LemmaNan06}} \label{AppG}
We will first prove (\ref{0713}). Applying (\ref{1711}), (\ref{eq:17}), (\ref{eq:18}) and (\ref{eq:36}), the following chains of inequalities can be written as
\begin{align}
	& M+R(M) 
	\nonumber\\
	\geq &H(Z_1)+H(X_1^{(1,2,2)})+H(X_2^{(1,2,2)})+H(X_3^{(1,2,2)}) \geq H(Z_1,X_2^{(1,2,2)},X_3^{(1,2,2)})+H(X_1^{(1,2,2)})
	\nonumber\\
	= &H(Z_1,X_2^{(1,2,2)},X_3^{(1,2,2)},W_1)+H(Z_2,Z_1,X_2^{(1,2,2)},W_1)-H(Z_2,Z_1,X_2^{(1,2,2)},W_1)
	\label{Nan40.99}\\
	&+H(X_1^{(1,2,2)})\nonumber\\
	\geq &H(X_3^{(1,2,2)},Z_2,Z_1,X_2^{(1,2,2)},W_1)+H(W_1,Z_1,X_2^{(1,2,2)})-H(W_1,Z_1,Z_2,X_2^{(1,2,2)})
	\label{Nan41}\\
	&+H(X_1^{(1,2,2)})\nonumber\\
	= &H(W_1,W_2)+H(W_1,Z_1,X_2^{(1,2,2)})-H(W_1,Z_1,Z_2)+H(X_1^{(1,2,2)})
	\label{Nan41.5}\\
	= &H(W_1,Z_1,X_3^{(1,2,2)})-H(W_1,Z_1,Z_2)+H(X_1^{(1,2,2)})+H(W_1,W_2)
	\label{Nan41.6}\\
	\geq &H(W_1,Z_1,Z_2,X_3^{(1,2,2)})-H(W_1,Z_1)+H(X_1^{(1,2,2)})+H(W_1,W_2)
	\label{Nan41.7}\\
	= &H(W_1,W_2)-H(W_1,Z_1,Z_2)+H(X_1^{(1,2,2)})+H(W_1,W_2)
	\label{Nan42}\\
	\geq &H(W_1,Z_1)-2H(W_1,Z_1,Z_2)+H(X_1^{(1,2,2)})+2H(W_1,W_2)
	\label{Nan42.1}\\
	= &H(W_1,Z_2)-2H(W_1,Z_1,Z_2)+H(X_1^{(1,2,2)})+2H(W_1,W_2)+H(W_1,X_1^{(1,2,1)})
	\nonumber\\
	&-H(W_1,X_1^{(1,2,1)})\nonumber\\
	\geq &H(W_1,Z_2,X_1^{(1,2,1)})-H(W_1,X_1^{(1,2,1)})-2H(W_1,Z_1,Z_2)+H(X_1^{(1,2,2)})
	\label{Nan43}\\
	&+2H(W_1,W_2)+H(W_1)\nonumber\\
	= &H(W_1,X_1^{(1,2,1)},X_2^{(1,2,1)},Z_2)+H(W_1,X_1^{(1,2,1)},X_2^{(1,2,1)},X_3^{(1,2,1)})
	\label{Nan43.5}\\
	& -H(W_1,X_1^{(1,2,1)},X_2^{(1,2,1)},X_3^{(1,2,1)}) -H(W_1,X_1^{(1,2,1)})-2H(W_1,Z_1,Z_2)+H(X_1^{(1,2,2)})
	\nonumber\\
	&+2H(W_1,W_2)+H(W_1)\nonumber\\
	\geq &H(W_1,X_1^{(1,2,1)},X_2^{(1,2,1)},X_3^{(1,2,1)},Z_2)+H(W_1,X_1^{(1,2,1)},X_2^{(1,2,1)})
	\label{Nan44}\\
	& -H(W_1,X_1^{(1,2,1)},X_2^{(1,2,1)},X_3^{(1,2,1)}) -H(W_1,X_1^{(1,2,1)})-2H(W_1,Z_1,Z_2)+H(X_1^{(1,2,2)})
	\nonumber\\
	&+2H(W_1,W_2)+H(W_1)\nonumber\\
	\geq &H(X_1^{(1,2,2)})+3H(W_1,W_2)+H(W_1)-H(W_1,X_1^{(1,2,1)},X_2^{(1,2,1)},X_3^{(1,2,1)})
	\label{eq:37}\\
	&-2H(W_1,Z_1,Z_2),\nonumber
\end{align}
where (\ref{Nan40.99}) follows from the fact that knowing $(Z_1, X_2^{(1,2,2)}, X_3^{(1,2,2)})$ can decode $W_1$, 
(\ref{Nan41}),  (\ref{Nan41.7}),  (\ref{Nan43}) and (\ref{Nan44}) all follow from the sub-modular property of the entropy function,
(\ref{Nan41.5}) and (\ref{Nan42}) both follow from the fact that $X_1^{(1,2,2)}$ is a deterministic function of $Z_1$, knowing $(Z_2, X_1^{(1,2,2)}, X_3^{(1,2,2)})$ can decode $W_2$, (\ref{1711}) and (\ref{eq:17}),
(\ref{Nan41.6}) follows from the property of user-index-symmetric schemes, (\ref{Nan42.1}) follows from the fact that $H(W_1,Z_1)\leq H(W_1,Z_1,Z_2)$,
(\ref{Nan43.5}) follows from the fact that $X_2^{(1,2,1)}$ is a deterministic function of $Z_2$,
(\ref{eq:37}) follows from the fact that knowing $(Z_2, X_1^{(1,2,1)}, X_3^{(1,2,1)})$ can decode $W_2$, (\ref{1711}), (\ref{eq:17}) and $H(W_1,X_1^{(1,2,1)},X_2^{(1,2,1)})\geq H(W_1,X_1^{(1,2,1)})$. 
Hence, (\ref{0713}) is proved.

Next, we will prove (\ref{0714}).
\begin{align}
	& 2M+R(M) \geq H(Z_2,X_1^{(1,2,1)},X_3^{(1,2,1)})+H(Z_1)+H(X_2^{(1,2,1)})
	\nonumber\\
	\geq &H(Z_2,X_1^{(1,2,1)},X_3^{(1,2,1)})+H(Z_1,X_2^{(1,2,1)})
	\nonumber\\
	= &H(W_2,Z_2,X_1^{(1,2,1)},X_3^{(1,2,1)},X_2^{(1,2,1)})+H(W_1,X_1^{(1,2,1)},X_2^{(1,2,1)},X_3^{(1,2,1)})
	\label{Nan44.99}\\
	& -H(W_1,X_1^{(1,2,1)},X_2^{(1,2,1)},X_3^{(1,2,1)})+H(Z_1,X_2^{(1,2,1)})
	\nonumber\\
	\geq &H(X_1^{(1,2,1)},X_2^{(1,2,1)},X_3^{(1,2,1)})-H(W_1,X_1^{(1,2,1)},X_2^{(1,2,1)},X_3^{(1,2,1)})
	\label{Nan45}\\
	&+H(Z_1,X_2^{(1,2,1)})+H(W_1,W_2)\nonumber\\
	= &H(X_1^{(1,2,1)},X_2^{(1,2,1)},X_3^{(1,2,1)})-H(W_1,X_1^{(1,2,1)},X_2^{(1,2,1)},X_3^{(1,2,1)})
	\label{Nan45.5}\\
	&+H(Z_1,X_2^{(1,2,1)},X_1^{(1,2,1)})+H(W_1,W_2)\nonumber\\
	\geq &H(X_1^{(1,2,1)},X_2^{(1,2,1)})+H(Z_1,X_1^{(1,2,1)},X_2^{(1,2,1)},X_3^{(1,2,1)})
	\label{Nan45.6}\\
	&-H(W_1,X_1^{(1,2,1)},X_2^{(1,2,1)},X_3^{(1,2,1)})+H(W_1,W_2)\nonumber\\
	= &H(X_1^{(1,2,1)},X_2^{(1,2,1)})+H(Z_1,X_2^{(1,2,1)},X_3^{(1,2,1)})
	\label{eq:38}\\
	&-H(W_1,X_1^{(1,2,1)},X_2^{(1,2,1)},X_3^{(1,2,1)})+H(W_1,W_2),\nonumber
\end{align}
where (\ref{Nan44.99}) follows from the fact that $X_2^{(1,2,1)}$ is a deterministic function of $Z_2$ and knowing $(Z_2, X_1^{(1,2,1)}, X_3^{(1,2,1)})$ can decode $W_2$,
(\ref{Nan45}) follows from the sub-modular property of the entropy function, (\ref{1711}) and (\ref{eq:17}),
(\ref{Nan45.5}) and (\ref{eq:38}) both follow from the fact that $X_1^{(1,2,1)}$ is a deterministic function of $Z_1$, 
(\ref{Nan45.6}) follows from the sub-modular property of the entropy function. Hence, (\ref{0714}) is proved.

Moreover, we directly have
\begin{align}
	3M+R(M) & \geq 2H(Z_1)+H(Z_1)+H(X_1^{(1,1,2)})+H(X_2^{(1,1,2)})+H(X_3^{(1,1,2)}) \nonumber\\
	& \geq 2H(Z_1)+H(Z_1,X_2^{(1,1,2)},X_3^{(1,1,2)})+H(X_1^{(1,1,2)}),\nonumber
\end{align}
Hence, (\ref{0715}) is proved.

\section{Proof of Lemma \ref{LemmaNan07}} \label{AppH}
We will first prove (\ref{0716}).
\begin{align}
	& 3M+2R(M) \geq 3H(Z_1)+2[H(X_1^{(1,1,2)})+H(X_2^{(1,1,2)})+H(X_3^{(1,1,2)})]
	\nonumber\\
	\geq &2H(Z_1,X_2^{(1,1,2)},X_3^{(1,1,2)})+H(Z_1,X_2^{(1,2,2)},X_3^{(1,2,2)})
	\label{Nan47}\\
	\geq &2H(W_1,Z_1,X_2^{(1,1,2)})+H(W_1,Z_1,X_3^{(1,2,2)})
	\label{Nan47.5}\\
	= &H(W_1,Z_1,X_2^{(1,1,2)})+H(W_1,Z_2,X_1^{(1,1,2)})+H(W_1,Z_1,X_3^{(1,2,2)})
	\label{Nan47.55}\\
	= &H(W_1,Z_1,X_2^{(1,1,2)},X_1^{(1,1,2)})+H(W_1,Z_2,X_1^{(1,1,2)},X_2^{(1,1,2)})+H(W_1,Z_1,X_3^{(1,2,2)})
	\label{Nan47.6}\\
	\geq &H(W_1,Z_1,Z_2,X_1^{(1,1,2)},X_2^{(1,1,2)})+H(W_1,X_1^{(1,1,2)},X_2^{(1,1,2)})+H(W_1,Z_1,X_3^{(1,2,2)})
	\label{Nan48}\\
	\geq &H(W_1,Z_1)+H(W_1,X_1^{(1,1,2)},X_2^{(1,1,2)})+H(W_1,Z_1,Z_2,X_1^{(1,1,2)},X_2^{(1,1,2)},X_3^{(1,2,2)})
	\label{Nan49}\\
	= &H(W_1,Z_1)+H(W_1,X_1^{(1,1,2)},X_2^{(1,1,2)})+H(W_1,W_2)
	\label{Nan49.5}\\
	= &2+H(W_1,Z_1)+H(W_1,X_1^{(1,1,2)},X_2^{(1,1,2)}),\nonumber
\end{align}
where (\ref{Nan47}) follows from (\ref{eq:36}), (\ref{Nan47.5}) follows from the fact that knowing $(Z_1, X_2^{(1,1,2)}, X_3^{(1,1,2)})$ can decode $W_1$ and knowing $(Z_1, X_2^{(1,2,2)}, X_3^{(1,2,2)})$ can decode $W_1$, 
(\ref{Nan47.55}) follows from the property of user-index-symmetric schemes, (\ref{Nan47.6}) follows from the fact that $X_1^{(1,1,2)}$ is a deterministic function of $Z_1$ and $X_2^{(1,1,2)}$ is a deterministic function of $Z_2$, (\ref{Nan48}) and (\ref{Nan49}) both follow from the sub-modular property of the entropy function, (\ref{Nan49.5}) follows from the fact that $X_1^{(1,2,2)}$ is a deterministic function of $Z_1$, knowing $(Z_2, X_1^{(1,2,2)}, X_3^{(1,2,2)})$ can decode $W_2$, (\ref{1711}) and (\ref{eq:17}). Hence, (\ref{0716}) is proved.

Next, we will prove (\ref{0717}). We notice that
\begin{align}
	&H(W_1,Z_1)+H(W_1,X_2^{(2,1,1)})+H(W_1,X_1^{(1,1,2)},X_2^{(1,1,2)})+H(W_1,X_1^{(1,1,2)},Z_3)
	\nonumber\\
	\geq &H(W_1,Z_1,X_2^{(2,1,1)})+H(W_1)+H(W_1,X_1^{(1,1,2)})+H(W_1,X_1^{(1,1,2)},X_2^{(1,1,2)},Z_3)
	\label{Nan50}\\
	= &H(W_1,Z_1,X_2^{(2,1,1)})+H(W_1)+H(W_1,X_1^{(1,1,2)})+H(W_1,X_1^{(1,1,2)},X_2^{(1,1,2)},Z_3,W_2)
	\label{Nan51}\\
	= &H(W_1,Z_1,X_2^{(2,1,1)})+H(W_1)+H(W_1,X_1^{(1,1,2)})+H(W_1,W_2)
	\label{Nan51.1}\\
	= &3+H(W_1,Z_1,X_2^{(2,1,1)})+H(W_1,X_1^{(1,1,2)})
	\nonumber\\
	= &3+H(W_1,Z_3,X_1^{(1,1,2)})+H(W_1,X_2^{(2,1,1)}),
	\label{Nan51.5}
\end{align}
where (\ref{Nan50}) follows from the sub-modular property of the entropy function, (\ref{Nan51}) follows from the fact that knowing $(Z_3, X_1^{(1,1,2)}, X_2^{(1,1,2)})$ can decode $W_2$, (\ref{Nan51.1}) follows from (\ref{1711}) and (\ref{eq:17}), (\ref{Nan51.5}) follows from the property of user-index-symmetric schemes.
Removing $H(W_1,Z_3,X_1^{(1,1,2)})+H(W_1,X_2^{(2,1,1)})$ from both sides of (\ref{Nan51.5}), (\ref{0717}) is proved.

\section{Proof of Theorem \ref{Theo3}} \label{0815}
	
	
	
	The corner point of $(N,0)$ is trivial. 
	As for the corner point of $\left(\frac{1}{s+1}N,\frac{s}{s+1}\min\{N,K-s\}\right)$, its achievability scheme is as follows: split all files into $s+1$ subfiles of equal sizes. For each file, use a $(K,s+1)$-MDS code to encode the $s+1$ subfiles. The encoded copies of File $n$ are denoted as $\bar{W}_{n,i}$, $n \in [N], i \in [K]$. User $k$ stores $\left\{\bar{W}_{n,k}, n \in [N] \right\}$, which is of size $\frac{1}{s+1}N$. 
	%
	%
	In the delivery phase, each sender in $\mathcal{S}$ transmits the encoded copy of the requested files stored in its cache. Denote the number of distinct requested files as $N_e$, then the transmission rate of each sender is $\frac{1}{s+1}N_e$, and the total transmission rate of all $s$ senders is $\frac{s}{s+1}N_e$. In the worst-case, $N_e=\min\{N,K-s\}$. Hence, the worst-case delivery rate is $R(M)=\frac{s}{s+1}\min\{N,K-s\}$. Due to the $(K,s+1)$-MDS code employed, each requester can decode its requested file based on the one encoded copy in its own cache and the $s$ encoded copies from the $s$ senders.

	%
	
	Lastly, we provide the achievability scheme for the corner point of $\left(\frac{K-1}{K}N,\frac{1}{K} \right)$. The caching scheme is the same as that of the MAN uncoded symmetric placement in \cite[Algorithm 1]{6620392} with $t=K-1$. More specifically, all files are split into $K$ subfiles of equal sizes, denoted as $(W_{n,1}, \cdots, W_{n, K})$, and User $k$ stores $\left\{(W_{n,1}, \cdots, W_{n,k-1}, W_{n,k+1},\cdots, W_{n,K}), n \in [N]\right\}$, which is of size $\frac{K-1}{K}N$. In the delivery phase, denote the set of users not requesting files as $\mathcal{S}$, and let User $k \in [K] \setminus \mathcal{S}$ requests File $d_k$. One of the designated senders in $\mathcal{S}$ sends the signal $\oplus_{k\in[K]\setminus \mathcal{S}}W_{d_k,k}$, which is of rate $\frac{1}{K}$. Note that the sender has all the pieces $W_{d_k,k},  k\in[K]\setminus \mathcal{S}$ in its cache. It is easily checked that the decodability constraint is satisfied. 
	
	
	%
	%

\ifCLASSOPTIONcaptionsoff
\newpage
\fi


\end{document}